\documentclass[twocolumn,prx,superscriptaddress]{revtex4}
\usepackage[latin9]{inputenc}
\usepackage{color}
\usepackage{amsmath}
\usepackage{amssymb}
\usepackage{graphicx}
\usepackage{float}
\usepackage{algorithm}
\usepackage{algorithmic}
\usepackage{url}
\usepackage{physics}
\usepackage[most]{tcolorbox}
\usepackage{amsthm}
\usepackage{bm}
\usepackage{bbm}
\usepackage{layout}
\usepackage{amsfonts}
\usepackage{amssymb}
\usepackage{mathtools}
\theoremstyle{definition}
\newtheorem{theorem}{Theorem}

\newtheorem{lemma}{Lemma}

\renewcommand{\O}[1]{O\left(#1\right)}
\renewcommand{\Re}{\text{Re}}

\usepackage[unicode=true,
 bookmarks=true,bookmarksnumbered=false,bookmarksopen=false,
 breaklinks=false,pdfborder={0 0 1},backref=false,colorlinks=true]
 {hyperref}
\hypersetup{
 linkcolor=magenta, urlcolor=blue, citecolor=blue, pdfstartview={FitH}, hyperfootnotes=false, unicode=true}

\usepackage[capitalise,compress]{cleveref}
\makeatletter
\@ifundefined{textcolor}{}
{%
 \definecolor{BLACK}{gray}{0}
 \definecolor{WHITE}{gray}{1}
 \definecolor{RED}{rgb}{1,0,0}
 \definecolor{GREEN}{rgb}{0,1,0}
 \definecolor{BLUE}{rgb}{0,0,1}
 \definecolor{CYAN}{cmyk}{1,0,0,0}
 \definecolor{MAGENTA}{cmyk}{0,1,0,0}
 \definecolor{YELLOW}{cmyk}{0,0,1,0}
}

\usepackage{amsfonts}
\usepackage{tabularx}
\usepackage{dcolumn}
\usepackage{bm}
\usepackage{epstopdf}
\usepackage{times}

\hypersetup{urlcolor=blue}

\renewcommand{\O}[1]{O\left(#1\right)}

\makeatother

\begin{document}

\title{Improved Digital Quantum Simulation by Non-Unitary Channels}

\author{W. Gong}
\affiliation{School of Engineering and Applied Sciences, Harvard University, Cambridge, Massachusetts 02134, USA}
\author{Yaroslav Kharkov}
\thanks{Currently at AWS Quantum Technologies}
\affiliation{Joint Center for Quantum Information and Computer Science, NIST/University of Maryland, College Park, Maryland 20742, USA}
\affiliation{Joint Quantum Institute, NIST/University of Maryland, College Park, Maryland 20742, USA}
\author{Minh C. Tran}
\affiliation{Joint Center for Quantum Information and Computer Science, NIST/University of Maryland, College Park, Maryland 20742, USA}
\affiliation{Joint Quantum Institute, NIST/University of Maryland, College Park, Maryland 20742, USA}

\author{Przemyslaw Bienias}
\affiliation{Joint Center for Quantum Information and Computer Science, NIST/University of Maryland, College Park, Maryland 20742, USA}
\affiliation{Joint Quantum Institute, NIST/University of Maryland, College Park, Maryland 20742, USA}
\author{Alexey V. Gorshkov}
\affiliation{Joint Center for Quantum Information and Computer Science, NIST/University of Maryland, College Park, Maryland 20742, USA}
\affiliation{Joint Quantum Institute, NIST/University of Maryland, College Park, Maryland 20742, USA}

\begin{abstract}
Simulating quantum systems is one of the most promising avenues to harness the computational power of quantum computers. However, hardware errors in noisy near-term devices remain a major obstacle for applications.   
Ideas based on the randomization of Suzuki-Trotter product formulas have been shown to be a powerful approach to reducing the errors of quantum simulation and lowering the gate count. 
In this paper, we study the performance of non-unitary simulation channels and consider the error structure of channels constructed from a weighted average of unitary circuits. 
We show that averaging over just a few simulation circuits can significantly reduce the Trotterization error for both single-step short-time and multi-step long-time simulations. 
We focus our analysis on two approaches for constructing circuit ensembles for averaging: (i) permuting the order of the terms in the Hamiltonian and (ii) applying a set of global symmetry transformations.
We compare our analytical error bounds to empirical performance and show that empirical error reduction surpasses our analytical estimates in most cases. 
Finally, we test our method on an IonQ trapped-ion quantum computer accessed via the Amazon Braket cloud platform, and benchmark the performance of the averaging approach.
\end{abstract}

\maketitle

\section{Introduction}

Near-term quantum computational devices possess the potential to simulate the dynamics of many-body quantum systems beyond the capability of classical computers~\cite{kim2023evidence}. 
The two main approaches for simulating quantum dynamics are (i) analog Hamiltonian simulation and (ii) digital quantum simulation.  
In digital quantum simulation, the physical system is usually mapped to qubits, and quantum gates are used to approximate the desired dynamics.
The most common class of algorithms are product formulas, including the Suzuki-Trotter formulas~\cite{Suzuki1985,Suzuki1990Fractal,Suzuki1991General,Huyghebaert1990,ChildsY19,Childs2004, BerryACS07,ChildsMNRS2017}, which approximate the evolution unitary by splitting the Hamiltonian into a sum of non-commuting terms and
evolving the quantum state with each term in a sequence of small time intervals.
More advanced algorithms, such as truncated Taylor series~\cite{BerryCCKS2015}, multiproduct formulas~\cite{Low19,zhuk2023trotter}, and methods based on quantum signal processing~\cite{LowC2017,Low2019Hamiltonian}, have also been developed in recent years. 
Nevertheless, product formulas remain popular, especially for near-term devices~\cite{Blatt2012Quantum,Monroe2021Programmable,Kjaergaard2020Superconducting,Weimer2010Rydberg,Weimer2011Digital,Lesanovsky2012Liquid,Schonleber2015Quantum,Bloch2012Quantum}, due to their relative implementation simplicity.   

\begin{figure}
    \centering
    \includegraphics[width=0.49\textwidth]{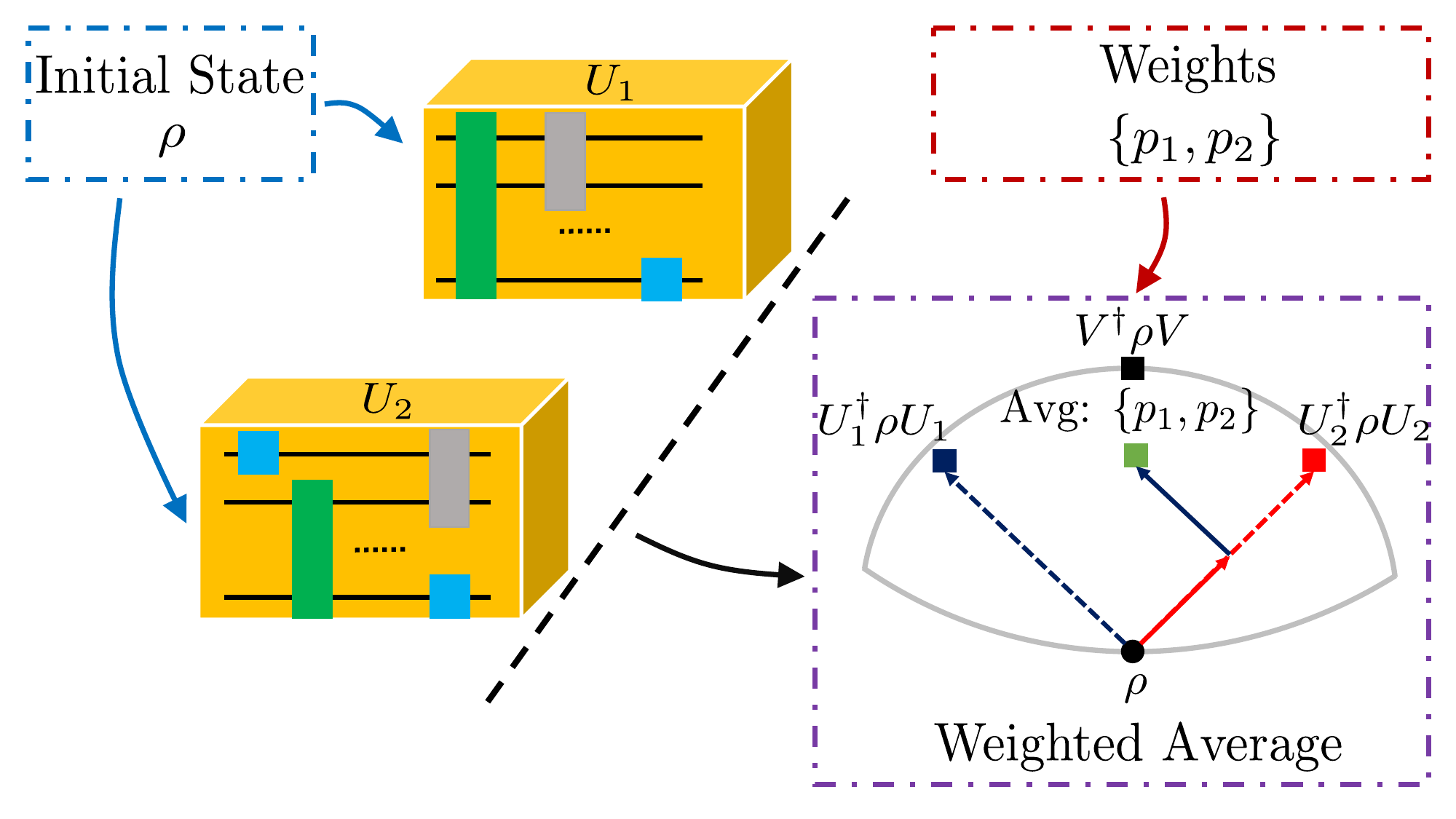}
    \caption{An illustration of the weighted average over unitary simulations and the intuition for error reduction. We prepare $M$ unitary circuits $U_1,U_2,...,U_M$ (each aiming to approximate the ideal evolution $V=e^{-iHt}$) and specify a set of weight coefficients $\{p_1,...,p_M\}$ for each circuit (the figure shows the case $M = 2$). 
    The initial quantum state $\rho$ is evolved under the application of individual unitary circuits $U_m$. In order to obtain the final measurement result, a weighted average of the measurement outputs is computed $\langle O \rangle_{p} = \sum_{m=1}^M p_m \Tr [O U^\dag_m \rho U_m]$, where $O$ is the observable of interest. 
    The average of simulations reduces the simulation error due to a partial cancellation of individual error terms.
    }
    \label{fig:Illu}
\end{figure} 

Recently, a new class of quantum simulation algorithms based on randomization has been proposed~\cite{Campbell2019Random,Childs2019Faster,Chen2021ConcentrationPRX,Hastings2017Turning,Cho2022Doubling,Ouyang2020Compilation,Faehrmann2022Randomizing,Nakaji2023Qswift}.
In contrast to coherent algorithms that produce unitary quantum circuits, these randomized algorithms generally result in non-unitary simulation channels by randomly choosing a product formula in each time step.
These non-unitary channels have better asymptotic scaling with certain Hamiltonian parameters than deterministic algorithms.

In this paper, we study the efficiency of non-unitary simulation channels in various non-asymptotic scenarios. 
Specifically, we consider non-unitary simulation channels (NUSCs) constructed from weighted averages of unitary simulation circuits (as shown in Fig.~\ref{fig:Illu}). 
We provide an estimate for the error structure of NUSCs and demonstrate how NUSCs can significantly reduce the simulation error in short-time and long-time simulations.
Here, the NUSCs for long-time simulation are constructed by repeatedly using the same circuit block for each time step, which is different from previous randomized simulation approaches~\cite{Childs2019Faster,Campbell2019Random,Chen2021ConcentrationPRX,Low2022Complexity}.
We consider a few methods for constructing contributing circuits and calculate the optimal distribution of weight coefficients of NUSCs in some cases. 
We demonstrate the performance of NUSCs for digital quantum simulation of physically relevant systems both numerically and experimentally on the IonQ Harmony device~\cite{Wright2019Benchmarking} via the Amazon Braket cloud platform~\cite{AmazonBraket}.

The rest of the paper is organized as follows. 
In \Cref{sec:GenFrame}, we introduce the general construction of NUSCs for digital quantum simulation and express the error of NUSCs in terms of the errors of the contributing unitary simulation circuits. 
In \Cref{sec:MixChanCon}, we provide two constructions for NUSCs and evaluate their performance both analytically and numerically. Both constructions average over simulation circuits based on Suzuki-Trotter product formulas. In the first construction (\cref{sec:TermOrdering}), the different circuits are obtained by permuting the order of the Hamiltonian terms. In the second construction (\cref{sec:TermOrdering}), the different circuits are obtained by symmetry transformations. 
We derive analytical representations for and bounds on the error for each construction of NUSCs. We benchmark the technique by applying it to a nearest-neighbor XY spin chain and to a Heisenberg chain with power-law interactions.
In \Cref{sec:BraketExp}, we benchmark the performance of NUSCs experimentally on an IonQ trapped-ion quantum computer. In \Cref{sec:outlook}, we provide the summary and the outlook. In the Appendices, we provide details omitted from the main text.

\section{The General framework}\label{sec:GenFrame}
We consider an $n$-qubit state $\rho$ and the task of simulating on a quantum computer the evolution $\mathcal{V}(\rho) = V \rho V^\dagger$, where $V = e^{-iHt}$ describes the time evolution under Hamiltonian $H$ for some time $t$. A quantum simulation of the evolution $\mathcal{V}(\cdot)$ is a quantum channel $\mathcal{U}(\cdot)$ that approximates the target evolution for any quantum state or a particular quantum state $\rho$. Assume we have $M$ unitary simulation channels $\{\mathcal{U}_1,\ldots,\mathcal{U}_M\}$ with the corresponding unitaries $\{U_1,\dots, U_M\}$. 
We consider averaging them according to a vector of $M$ non-negative weights $\{p_1,...,p_M\}$ such that $\sum_{m=1}^M p_m = 1$ (as shown in Fig.~\ref{fig:Illu}). Our aim is to compare how well the NUSC
\begin{align}\label{eq:Mixture}
\mathcal{U}(\rho) = \sum_{m = 1}^M p_m\mathcal {U}_m(\rho) = \sum_{m = 1}^M p_mU_m \rho U_m^\dag
\end{align}
approximates $\mathcal V(\rho)$ compared to each contributing simulation $\mathcal{U}_m(\rho)$. In the following, we will use the Frobenius norm $\norm{A}_F=\Tr(A^\dagger A)$ for convenience in the analytical calculations. It is trivial to show that
\begin{align}\label{eq:MixErrBound}
\norm{\mathcal U(\rho) - \mathcal V(\rho)}_F 
\leq \sum_{m} p_m \norm{\mathcal U_m(\rho) - \mathcal V(\rho)}_F.
\end{align}
Therefore, the performance of $\mathcal U(\rho)$ cannot be worse than the (weighted) average performance of the contributing simulations.
In particular, if $\norm{\mathcal U_m(\rho) - \mathcal V(\rho)}_F = \epsilon$ for all $m$, the weighted average $\mathcal U(\rho)$ approximates the target state $\mathcal V(\rho)$ at least as well as each $\mathcal U_m$ does. 

In practice, the performance of $\mathcal{U}$ can be much better than guaranteed by \cref{eq:MixErrBound}.
To better reflect the typical performance of $\mathcal U(\rho)$, we consider the error of $\mathcal U$ averaged over Haar-random pure initial states~\cite{Ratcliffe1994Foundations} as the figure of merit (which we call the loss function):
\begin{equation}\label{eq:GenLoss}
L_F\equiv \int d\psi \norm{\mathcal U(\ket{\psi}\bra{\psi})-\mathcal V(\ket{\psi}\bra{\psi})}_F^2.
\end{equation}
We note that the choice of the Frobenius norm (as opposed to other metrics, e.g.~the spectral norm) does not qualitatively change our conclusions, but the Frobenius norm enables exact analytical calculations of the integral above. 

We consider a short evolution time $t$ such that $\norm{H}_Ft\ll1$ and assume that $U_m$ and $V$ are functions of $t$ and are both infinitely differentiable. Under this assumption, we consider the Taylor expansion of the error of each $U_m$  in powers of $t$. 
Suppose the leading order error in each contributing simulation $U_1,...,U_M$ is $O(t^q)$ for a positive integer $q$. We can then decompose each $U_m$ up to order $t^{2q}$ as
\begin{align}\label{eq:Umseries}
U_m=V-i\sum_{s=q}^{2q}E_{m}^{(s)}t^s+O(t^{2q+1}),
\end{align}
where $E_{m}^{(s)}=\frac{i}{s!}\frac{d^s}{dt^s}(U_m-V)$ is the coefficient matrix of the $s$th-order error term in the Taylor expansion of $U_m$.  The following result relates the error of the weighted average $\mathcal U$ to the error terms $E_{m}^{(s)}$ of the contributing simulations $\mathcal U_m$. We consider $\sqrt{L_F}$, rather than $L_F$ itself, in the following theorem so that the error appears at the same order $q$ as in contributing simulations $\mathcal{U}_m$. 

\begin{theorem}\label{thm:average-error}
Given the expansions of $U_m$ in \cref{eq:Umseries} and a set of weights $\{p_1,...,p_M\}$ normalized to unity, the error of the weighted average $\mathcal U$ can be written as
\begin{equation}\label{eq:FrobLossSameOrd}
\sqrt{L_F}=\sqrt{\frac{2}{d(d+1)}}\left(L_0t^{q}+L_{1}t^{q+1}\right)+O(t^{q+2}),
\end{equation}
where $d=2^n$ is the dimension of the Hilbert space and
\begin{align}\label{eq:FrobLossSameOrd2q}
L_0^2&=d\norm{E_q}_F^2-\abs{\Tr(E_q^\dagger V)}^2,\\\label{eq:FrobLossSameOrd2q1}
L_{1}^2&=\frac{\abs{d\Tr(E_{q+1}^\dagger E_q)-\Tr(E_{q+1}^\dagger V)\Tr(V^\dagger E_q)}}{L_0},
\end{align}
with the weighted $q$th-order coefficient matrix defined as $E_{q}\equiv\sum_{m=1}^M p_m E_{m}^{(q)}$.
\end{theorem}

We provide detailed proof for the theorem in \Cref{section:ProofThm1}. Although it is hard to obtain an analytical optimal probability that minimizes Eq.~\eqref{eq:FrobLossSameOrd}, the expressions for the leading ($\propto t^q$) and the second leading ($\propto t^{q+1}$) order error terms in \Cref{thm:average-error} relate the error of the NUSC $\mathcal{U}(\cdot)$ to the errors of the contributing unitary simulations $\mathcal{U}_m$. 
In particular, \cref{eq:FrobLossSameOrd2q} shows how a reduction in $E_q$, the weighted $q$th-order error, reduces the error of the NUSC. 
Comparing the loss functions for each contributing simulation and for the NUSC given by \cref{thm:average-error}, one can see that the simulation error of $\mathcal U$ can be suppressed by choosing a set of weights that result in smaller $E_q$ and $E_{q+1}$. For technical simplicity, we use the error reduction from $E_m^{(q)}$'s to $E_q$ (for the leading order $q$) as a benchmark for the error reduction by a NUSC compared to each contributing circuit. For example, suppose we have a Hamiltonian $H=A+B$ with two simulations $U_1=e^{-iAt}e^{-iBt}$ and $U_2=e^{-iBt}e^{-iAt}$, each having a nonzero second-order error. If one sets the weights to be $(p_1,p_2)=(0.5,0.5)$, the second-order error term $E_2$ will vanish. Therefore, the average error of the NUSC is one order in $t$ better than the error of each contributing simulation. In general, however, for a small number $M$ of contributing simulations, matrices $E_q$ and $E_{q+1}$ cannot be reduced to zero by tuning the weight coefficients $p_m$ because the error matrices contain exponentially many (say, in the number of qubits being simulated) entries, while we only have $M-1$ tuning parameters. Nevertheless, as in the simple $H = A+B$ example above, we can expect a significant reduction of one or both these terms for carefully crafted weights and contributing simulations.

For long-time simulations with $\norm{H}_Ft\gg 1$, a standard approach in digital quantum simulation is to divide the evolution time into $N$ steps of duration $\Delta t$, such that  $\norm{H}_F\Delta t=\norm{H}_Ft/N\ll 1$. One then applies the same product formula in each time step, yielding a unitary simulation circuit. We first consider the error structure of such an $N$-step simulation, which is denoted as $U(N,\Delta t)$. We also denote the exact evolution and product formula in each time step as $V(\Delta t)=e^{-iH\Delta t}$ and $U(\Delta t)$, respectively. The simulation error in each step is $E(\Delta t)=V(\Delta t)-U(\Delta t)$. We decompose the error $E(\Delta t)$ into two parts as
\begin{align}\label{eq:ErrDecomp}
E(\Delta t)=[H,\eta(\Delta t)]+\xi(\Delta t).
\end{align}
The first part $[H,\eta(\Delta t)]$ represents the error that does not commute with the Hamiltonian, while the second part $\xi(\Delta t)$ commutes with $H$, i.e., $[H,\xi(\Delta t)] = 0$. The error for the $N$-step simulation $E(N,\Delta t)$ can be approximated to leading order as
\begin{align}\label{eq:LongTimeErrorRep}
E(N,\Delta t)\approx N\xi(\Delta t)+\frac{1}{\Delta t}[e^{-iHt}\eta(\Delta t)e^{iHt}-\eta(\Delta t)].
\end{align}
In \Cref{section:ErrDecomp}, we provide the proof of \cref{eq:LongTimeErrorRep} and more details on how to derive $\xi(\Delta t)$ and $\eta(\Delta t)$ from $H$ and $E(\Delta t)$. When the step size $\Delta t$ is fixed, the norm of the second term in \cref{eq:LongTimeErrorRep} is bounded by a constant $O(\norm{\eta(\Delta t)}_F)$ while the first term increases linearly with $N$. Therefore, the first term becomes the dominant error in the long-time (large-$N$) limit.

In order to study the general structure of the NUSC error in long-time simulation, we proceed by analogy with the setting for a short evolution time. We suppose that the leading-order term in $E_m(\Delta t)$ is $O(\Delta t^q)$ for an integer $q$. We can thus write the expansion of $E_m(\Delta t)$ similarly to \cref{eq:Umseries} as $E_m(\Delta t)=V(\Delta t)-U_m(\Delta t)=iE_m^{(q)}\Delta t^q+O(\Delta t^{q+1})$. According to \cref{eq:ErrDecomp}, we split the coefficient matrix $E_m^{(q)}$ into two terms as
\begin{align}\label{eq:UmErrDecomp}
E_m^{(q)}=[H,\eta_m^{(q)}]+\xi_m^{(q)},
\end{align}
where $[H,\xi_m^{(q)}]=0$. We can decompose the $N$-step simulation $U_m(N,\Delta t)$ up to order $\Delta t^q$ and ignore the part that does not increase with $N$:
\begin{align}\label{eq:LongUmExpansion}
U_m(N,\Delta t)=V-iN\xi_m^{(q)}\Delta t^q +O(N\Delta t^{q}),
\end{align}
where $V=e^{-iHt}=e^{-iHN\Delta t}$. We then combine \cref{eq:LongUmExpansion} with \cref{thm:average-error} and obtain the following error estimate:
\begin{align}\label{eq:CoroLongTine}
\sqrt{L_F}=\sqrt{\frac{2}{d(d+1)}}N[L_0^{\text{long}}\Delta t^q+O(\Delta t^{q+1})],
\end{align}
where $(L_0^{\text{long}})^2=d\norm{\xi_q}_F^2-\abs{\Tr(\xi_q^\dagger V)}^2$ and $\xi_q=\sum_{m=1}^Mp_m\xi_m^{(q)}$. \cref{eq:CoroLongTine} shows how a reduction in $\xi_q$, 
the weighted $q$th-order commuting error,
reduces the simulation error of the NUSC. By carefully choosing weights $\{p_1,\ldots,p_M\}$, we can expect a significant reduction of $\xi_q$ compared to the contributing terms $\xi_m^{(q)}$, which results in a smaller simulation error for the NUSC.

\section{Constructions of NUSCs}
\label{sec:MixChanCon}
 In this section, we consider two constructions for each step in the contributing unitaries: by permuting the terms of the Hamiltonian~\cite{Childs2019Faster} (\Cref{sec:TermOrdering}) and by symmetry transformations~\cite{Tran2021Faster} (\Cref{sec:Sym}).
\subsection{Permutation of Hamiltonian term ordering}
\label{sec:TermOrdering}
We construct contributing simulations by using product formulas with different term orderings.  We start with short-time simulations when $\norm{H}_F t\ll 1$. To construct the contributing simulation unitaries for the weighted combination, we can exploit a single-step product formula. Suppose the Hamiltonian can be decomposed into $\Gamma$ experimentally realizable non-commuting terms as $H=\sum_{j=1}^\Gamma H_j$. 
Additionally, we also assume the Hamiltonian can be parametrized as a sum of Pauli string terms
\begin{align}\label{eq:klocalPauliH}
    H = \sum_{\bm \sigma} J_{\bm \sigma} \bm \sigma,
\end{align}
where $\bm \sigma = \sigma_1\otimes\dots\otimes\sigma_n$ denotes Pauli strings on $n$ qubits and $\sigma_i\in \{I,X,Y,Z\}$ for all $1\leq i\leq n$ denote Pauli matrices on site $i$. In the following, we also use $X_i$ ,$Y_i$, and $Z_i$ to denote Pauli $X$, $Y$, and $Z$ operators on site $i$. Multiple Pauli terms $J_{\bm \sigma} \bm \sigma$ are allowed to be included in a single $H_j$ Hamiltonian term.

\begin{figure*}
    \centering
    \includegraphics[width=0.98\textwidth]{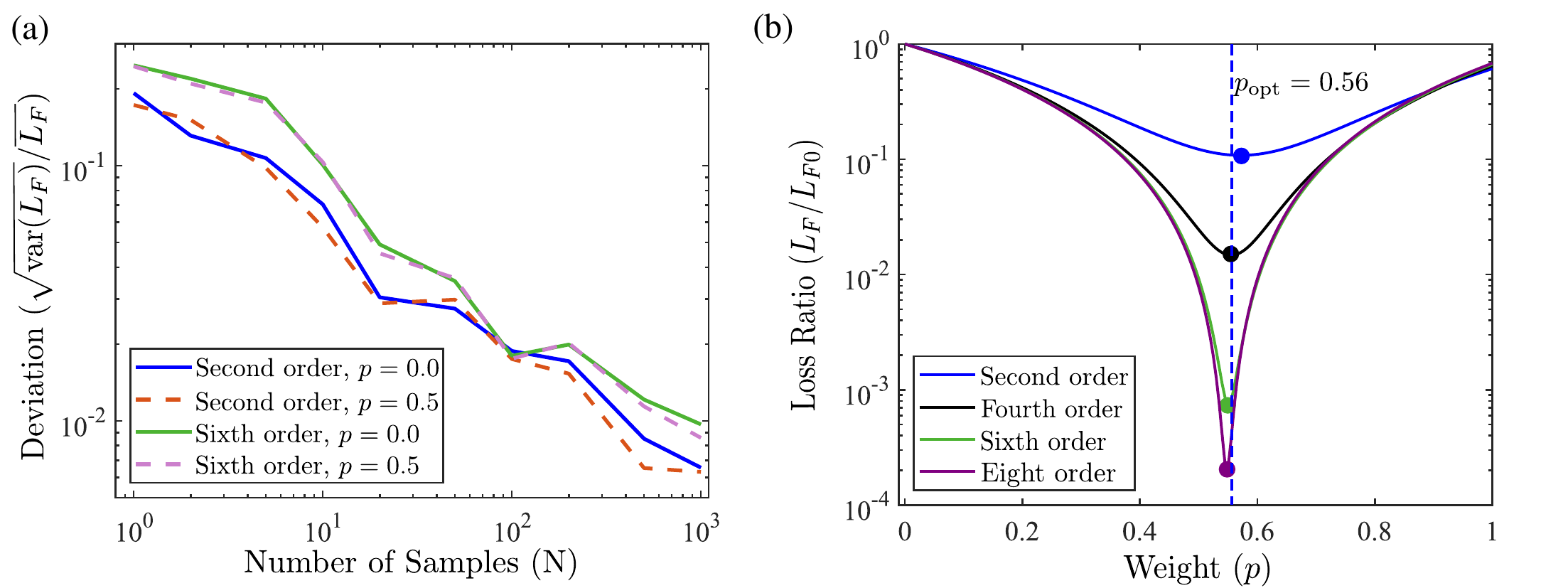}
    \caption{(a) Deviation of the empirical average from the actual loss function. We consider the short-time simulation of a nearest-neighbor XY spin chain in Eq.~\eqref{eq:1DHeisenbergLoc} at $t=0.3$.  For the second-order ($k=2$) Suzuki-Trotter formula, we use $U_1$ and $U_2$ given in Eqs.\  (\ref{eq:suzukiini1},\ref{eq:suzukiini2}). For $k= 4$, $6$, and $8$, we construct $U_1$ and $U_2$ using Eq.\ (\ref{eq:suzuki}) with $(A,B)$ and $(B,A)$ orderings. We plot the deviation of $L_F$, quantified using the ratio between the root-mean-square deviation $\sqrt{\text{var}(L_F)}$ and the averaged loss function $\overline{L_F}$ over $20$ batches of samples, as a function of sample number $N$ in each batch at different $k$'s and $p$'s. (b) Error reduction by averaging two Suzuki-Trotter formulas with $(A, B)$ and $(B, A)$ orderings of Hamiltonian terms. We plot the error metric $L_F/L_{F0}$ as a function of the weight coefficient $p$, where $L_{F0}$ is the simulation error for Suzuki-Trotter formulas $U_1$ of $k$th-order. The dashed vertical line shows the optimal weight $p$ obtained in \cref{thm:opt-weight} for second-order Suzuki-Trotter product formulas. The circles mark the minima of the curves and correspond to the optimal $p$ for each $k$. We see that these optimal $p$'s depend very weakly on $k$.
    }
    \label{fig:STcomb}
\end{figure*}

We consider a digital quantum simulation that approximates the target evolution through a $k$th-order Suzuki-Trotter product formula:
\begin{align}\label{eq:suzuki2}
&S_1(t)=\prod_{j=1}^\Gamma\text{exp}(-iH_jt),\\
&S_2(t) = \prod_{j=1}^\Gamma\text{exp}\left(-iH_j\frac{t}{2}\right)\prod_{j=\Gamma}^1\text{exp}\left(-iH_j\frac{t}{2}\right),\\\label{eq:suzuki}
&S_{k}(t) =  S^2_{k-2}\left(u_{k}t\right)S_{k-2}\left((1-4 u_{k})t\right)S_{k-2}^2\left(u_{k}t\right),
\end{align}
where $u_{k}=(4-4^{1/(k-1)})^{-1}$. The $k$th-order Suzuki-Trotter formula requires $O(5^{k/2})$ simulation blocks and provides a simulation approximation error $\epsilon_k=O(t^{k+1})$. 
Obviously, one way to reduce the approximation error is to increase the order $k$ of the Suzuki-Trotter product formula. However, this comes with the cost of an exponential increase in the number of gates in the simulation circuit. In contrast, in our approach based on NUSCs, we keep the order $k$ of the product formula (and hence the gate count) fixed.   
By building upon Suzuki-Trotter product formulas, various effective quantum simulation algorithms have been proposed \cite{Tran2021Faster,Campbell2019Random,Childs2019Faster} which improve the simulation error. 

We focus on near-term applications of digital quantum simulation with Hamiltonians consisting of a few terms.
In the following analysis, we consider simulating a Hamiltonian $H$ of the form in \cref{eq:klocalPauliH} that can be divided into two parts, i.e.~$\Gamma=2$. We simulate the evolution under $H$ by averaging the second-order product formulas with different term orderings.
Specifically, let $S$ be a subset of Pauli strings and let us decompose the Hamiltonian as a sum of two terms $H=A+B$, such that $A = \sum_{\bm \sigma \in S} J_{\bm \sigma}\bm \sigma$ and $B = \sum_{\bm \sigma \in S^c} J_{\bm \sigma}\bm \sigma$,
where $S^{c}$ is the complement of $S$.  

We consider the simulation of $\exp(-iHt)$ using second-order Suzuki-Trotter formulas \eqref{eq:suzuki2} with two possible term orderings:
\begin{align}
&U_1(t) =\text{exp}\left(-iA\frac{t}{2}\right)\text{exp}\left(-iBt\right)\text{exp}\left(-iA\frac t2\right),\label{eq:suzukiini1}\\
&U_2(t) =\text{exp}\left(-iB\frac{t}{2}\right)\text{exp}\left(-iAt\right)\text{exp}\left(-iB\frac t2\right).\label{eq:suzukiini2}
\end{align}
In the limit $\norm{H}_F t\ll 1$, \cref{eq:suzukiini1,eq:suzukiini2} both approximate $\exp(-iHt)$ up to the second order in $t$.
When averaging these two formulas with weights $\{p,1-p\}$, \cref{thm:average-error} relates the error of the resulting NUSC to the weighted third-order error $E_3$ defined after \cref{eq:FrobLossSameOrd2q1}.
The following lemma gives an expression for the norm of $E_3$.

\begin{lemma}\label{thm:opt-weight}
The norm of the weighted third-order error $E_3$ in averaging \cref{eq:suzukiini1,eq:suzukiini2} with weights $p$ and $1-p$, respectively, is
\begin{align}\label{eq:2ordErrRep}
\norm{E_{3}}_F^2 = C 2^n \left[  (2-3p)^2 J_A^2 + (1-3p)^2 J_B^2\right],
\end{align}
where $C$ is a constant that may depend on $H$ but is independent of $p$, $J_A^2=\sum_{\bm \sigma \in S}J_{\bm \sigma}^2$, and $J_B^2=\sum_{\bm\sigma\in S^c}J_{\bm \sigma}^2$. 
In particular, the optimal $p$ that minimizes $\norm{E_3}_F$ is
\begin{equation}\label{eq:2ordBestP}
p_{\text{opt}}=\frac{2J_A^2+J_B^2}{3\left(J_A^2+J_B^2\right)}.
\end{equation}
The simulation error term $\norm{E_3}_F^2$ of the NUSC is $C2^n J_A^2J_B^2/(J_A^2+J_B^2)$ at $p_{\text{opt}}$.
\end{lemma}

In \Cref{section:ProofThm2}, we provide a proof of \Cref{thm:opt-weight}. The proof follows from the calculation of the third-order error terms in $U_1$ and $U_2$. We note that we can obtain analytic optimal weights using orthogonal bases other than the basis of Pauli strings. In the case where the Hamiltonian terms $A$ and $B$ are not orthogonal in the Pauli basis, the expression for the optimal weight (\ref{eq:2ordBestP}) will be modified, although the analytical calculation is straightforward. 

To benchmark the effectiveness of averaging product formulas, we apply this technique to a nearest-neighbor XY spin chain:
\begin{equation}\label{eq:1DHeisenbergLoc}
H=\underbrace{\sum_{i=1}^{n-1}X_iX_{i+1}+h\sum_{i=1}^nX_i}_{=A}+\underbrace{\sum_{i=1}^{n-1}Y_iY_{i+1}}_{=B},
\end{equation}
where $h$ is the strength of the magnetic field along the $\hat{x}$ direction.
In all numerical simulations that follow, we estimate $L_F$ by calculating the averaged $L_F$ on $N = 10^3$ Haar-randomly chosen quantum states. 
Here, the deviation of the empirical average from the actual Haar average is bounded by $\sim\mathcal{O}(1/\sqrt{N})$ according to the Chebyshev inequality, where the coefficient in front of $1/\sqrt{N}$ depends on the error (maximized over all possible initial states) between the target evolution and the simulation. Although this coefficient depends on the system size, the system size is small in the following, so we can ignore this dependence. 
We simulate $e^{-iHt}$ for $n = 6$, $h=1$, and $t = 0.3$ using an average of the two second-order product formulas given in \cref{eq:suzukiini1,eq:suzukiini2} with the corresponding probabilities $p$ and $1-p$. In \cref{fig:STcomb}(a), we plot the deviation of $L_F$ as a function of sample number $N$ in each batch at $k=2$, $k=6$, $p=0$, and $p=0.5$ to support the claim that the deviation is negligible compared to the actual loss for $N=10^3$ samples.
In \cref{fig:STcomb}(b), we plot the reduction ratio of the error, quantified using the loss function in \cref{eq:GenLoss},  as a function of $p$.
The optimal average of $U_1$ and $U_2$ corresponds to about an order of magnitude reduction in the error.
The vertical dashed line in \cref{fig:STcomb}(b) corresponds to the value of $p_\text{opt}=0.5625$ from \cref{thm:opt-weight}.
We note that $\norm{E_3}_F$ given in \cref{thm:opt-weight} is only indirectly related to the error of the NUSC through \cref{thm:average-error}.
Additionally, higher-order terms also contribute to the error of the NUSC.
Nevertheless, \cref{fig:STcomb}(b) shows an excellent agreement between the estimate from \cref{thm:opt-weight} and the true value of $p$ that minimizes the error of the NUSC.
\begin{figure*}
    \centering
    \includegraphics[width=0.98\textwidth]{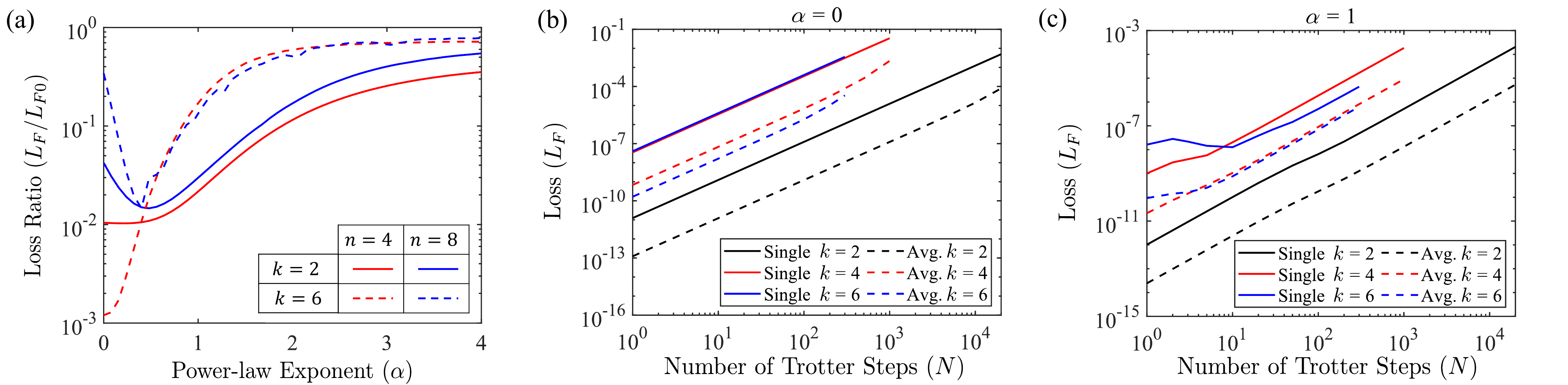}
    \caption{Long-time simulations of the power-law Heisenberg model in Eq.~\eqref{eq:PowerHeisenberg} using a combination of six Suzuki-Trotter formulas corresponding to all possible permutations of the Hamiltonian terms $H_X$, $H_Y$, $H_Z$. (a)  We plot the ratio of the figure of merit $L_F$ and the error $L_{F0}$ of  the contributing product formulas as a function of the power-law exponent $\alpha$ for the number of Trotter steps $N=2\times 10^4$ and the evolution time $t=100$ for different system sizes $n=4$, $8$. (b) In the $6$-qubit Heisenberg model (i.e.~$n=6$) with uniform all-to-all interactions ($\alpha=0$), we employ $2$nd-, $4$th-,  and $6$th-order Suzuki-Trotter decompositions for each step with fixed step sizes $0.005$, $0.1$, and $0.33$, respectively. We plot the loss function $L_F$ for the averaged channel  (the dotted lines) and the loss function for a single contributing formula (the solid lines) as a function of step number $N$. Since we fix the simulation time to be $t = 100$ but choose different step sizes for different $q$,  we get different maximal values of $N$ for different $q$. (c) The same as (b), but for  $\alpha=1$.}
    \label{fig:STcomblong}
\end{figure*}

For higher-order Suzuki-Trotter formulas, the leading error term becomes more complicated.
Therefore, it is more difficult to obtain a version of \cref{thm:opt-weight} for a general $k$th-order product formula.
Nevertheless, we will now argue that, for a non-degenerate (i.e.\ having all eigenenergies distinct) Hamiltonian with $\Gamma=2$, the 
reduction of the error term in the NUSC averaging $k\geq 2$th-order Suzuki-Trotter formulas become more prominent as $k$ grows.

To prove the above claim, we derive a recursive formula that approximately relates the error of the $k$th-order Suzuki-Trotter formula to the $(k-2)$th-order one. This recursive formula will allow us to estimate the error term $E_{k+1} $ for an arbitrary order $k$ of the Suzuki-Trotter formula.
For a $k$th-order Suzuki-Trotter formula $S_{k}$ (for either ordering of $A$ and $B$), we denote $E_{S_k}(t)=V(t)-S_k(t)$ as the simulation error. In the limit of short evolution time and large $k$, we have the following relationship between $E_{S_k}(t)$ and $E_{S_{k-2}}(t)$ (see \Cref{section:ProofInduct} for the proof):
\begin{align}\label{eq:STHighAppr}
E_{S_k}(t) =&\left[e^{iHu_kt},\left[e^{iH(1-2u_k)t},E_{S_{k-2}}(u_kt)\right]\right]  \nonumber\\ &+O((u_k-1/3)t^{k+1} ),
\end{align}
where $u_k=(4-4^{1/(k-1)})^{-1}$. 
Keeping only the linear part in the Taylor-series expansion of the exponentials, we get 
\begin{align}\label{eq:STHighApprSimp}
    E_{S_k}(t)&\approx (1-2u_k)u_kt^2\left[H,\left[H,E_{S_{k-2}}(u_kt)\right]\right]\nonumber\\
    &\approx(1-2u_k)u_k^kt^2\left[H,\left[H,E_{S_{q-2}}(t)\right]\right].
\end{align}

In the following, we consider combining $M$ $k$th-order Suzuki-Trotter formulas. To distinguish between different formulas, we denote by $E_{m}^{(k+1)}$ the $(k+1)$th-order error of the $m$th $k$th-order Suzuki-Trotter formula, i.e.~the $(k+1)$th-order term of $E_{S_k}$ for the $m$th $S_k$. Since $E_{k+1}=\sum_{m=1}^Mp_mE_{m}^{(k+1)}$, which is defined in Sec.\ \ref{sec:GenFrame} as the coefficient matrix for the averaged $(k+1)$th order error, is a linear combination of $E_{m}^{(k+1)}$, we can use Eq.~\eqref{eq:STHighApprSimp} to approximate $E_{k+1}$ using the Hamiltonian and $E_{k-2}$ or lower-order $E_{k'}$ ($k'<k-2$) for the same set of weights $\{p_1,...,p_M\}$. An important application is that we can use \cref{eq:STHighApprSimp} repeatedly to approximate $E_{k+1}(t)$ with $E_{3}(t)$ or $E_{5}(t)$ of the second-order or the $4$th-order Trotterization, which have analytical representations \cite{Childs2021Theory}. Specifically, we can combine \cref{eq:STHighApprSimp} and the error decomposition in \cref{eq:ErrDecomp}. Only the first term in \cref{eq:ErrDecomp} for $E_{k-2}$ will remain in $E_k$. Specifically, for a non-degenerate Hamiltonian (i.e.~Hamiltonian whose eigenvalues are all distinct) with $\Gamma = 2$, we have the following lemma regarding the error reduction when we average two higher-order Suzuki-Trotter formulas. 

\begin{lemma}\label{clm:ShortABErrRed}
For a non-degenerate Hamiltonian $H=A+B$, consider averaging two $k$th-order Suzuki-Trotter formulas  [defined in \cref{eq:suzuki}] in the small-$t$ limit [$\norm{H}_Ft\ll 1$] corresponding to the two orderings of $A$ and $B$, with weights $p$ and $1-p$, respectively. Then 

(i) The error reduction ratio $R_k=\norm{E_{k+1}}/\norm{E_{S_k}}$ for $k=4$ is larger than for $k=2$. Here $\norm{E_{S_k}}$ refers to an error for a fixed choice of the ordering of Hamiltonian terms in the Suzuki-Trotter formula. 

(ii) The optimal weight for averaging two $k$th-order formulas in the limit $k\to \infty$ approaches a fixed value $p_{\text{opt}}\to p_{\text{opt}}^*$.

\end{lemma}

We leave the detailed proof to \Cref{section:ErrDecomp}, where we use \cref{eq:STHighApprSimp} and the error decomposition in \cref{eq:ErrDecomp} to derive the coefficient matrix of the leading error term for combining two $k$th-order Suzuki-Trotter formulas with different orderings of $A$ and $B$. In the limit of large $k$, we further provide the conditions for completely eliminating the leading-order error. As a numerical illustration, we further compare the performance improvement of the combination of two $k$th-order Suzuki-Trotter formulas $S_k(t)$ for $k=2$, $4$, $6$, $8$ in Fig.~\ref{fig:STcomb}. We first observe that the optimal weights $p_{\text{opt}}$ remain nearly the same for $k=2$ and $k=4$. (They are not exactly the same because Lemma \ref{clm:ShortABErrRed} guarantees equality of optimal weights only at infinite $k$ and because we empirically optimize for  the loss function $L$ instead of the error term in the Lemma.) As the order increases, the improvement of the weighted average also increases. For two $8$-th order Suzuki-Trotter formulas, we find that NUSCs reduce the error metric $L_F$ by four orders of magnitude. These observations are consistent with the analytical results given by \Cref{clm:ShortABErrRed}.

For a general Hamiltonian $H=\sum_{\gamma=1}^\Gamma H_\gamma$ with $\Gamma\geq 3$, the error term $E_{m}^{(3)}$ depends heavily on the structure of $H$ and of each term $H_\gamma$. However, we can still expect an error reduction by averaging over different product formulas constructed by permuting the order of the terms $H_\gamma$. For example, consider the simplest case of the second-order product formula, $k=2$. The lowest-order simulation error for a second-order Suzuki-Trotter formula can be represented as \cite{Childs2021Theory}
\begin{align}\label{eq:2orderErrText}
\frac{t^3}{24}\sum_{\substack{\gamma_1=1\\ \gamma_2,\gamma_3=\gamma_1+1}}^\Gamma\left[H_{\gamma_1}+2 H_{\gamma_2},\left[H_{\gamma_3},H_{\gamma_1}\right]\right] + \O{t^4}.
\end{align}
Notice that there are at most $O(\Gamma^3)$ terms in \cref{eq:2orderErrText}, and the Frobenius norm for each commutator is bounded by $\norm{\left[H_{\gamma_1},\left[H_{\gamma_2},H_{\gamma_3}\right]\right]}_F\leq 4\norm{H_{\gamma_1}}_F\norm{H_{\gamma_2}}_F\norm{H_{\gamma_3}}_F$. Therefore, the simulation error for a single second-order Suzuki-Trotter formula is bounded above by $\norm{E_{S_2}}_F\leq O((t\Gamma\max_{\gamma=1,...,\Gamma}\norm{H_\gamma}_F)^3)$. Consider now averaging over all $\Gamma!$ second-order Suzuki-Trotter formulas corresponding to all possible orderings of $H_1,...,H_\Gamma$ with equal weights. The simulation error is now a linear combination of terms as in \cref{eq:2orderErrText} corresponding to all  $\Gamma!$ possible permutations of $H_\gamma$. We can observe that all commutators of the form $\left[H_{\gamma_1},\left[H_{\gamma_2},H_{\gamma_3}\right]\right]$, with $\gamma_1\neq\gamma_2\neq\gamma_3$,  will cancel each other. As a result, only at most $O(\Gamma^2)$ terms will remain, and the asymptotic simulation error will be bounded by $O((t\max_{\gamma=1,...,\Gamma}\norm{H_\gamma}_F)^3\Gamma^2)$, which is smaller compared to the contributing Trotterizations by a factor of $\Gamma$. The weighted average may also bring a further considerable reduction in the simulation error if we further optimize the weights or increase the order of the formulas. Note that, for a large number of terms $\Gamma$, there are exponentially many possible orderings of $H_\gamma$. We can therefore exploit random sampling techniques. In \Cref{section:ConvRandSam}, using matrix concentration theorems \cite{Tropp2015Introduction,Tropp2011Freedman,Gross2011Recovering,Chen2021Concentration,Chen2021ConcentrationPRX} and the spectral norm for technical simplicity, we prove that, with high probability,  the statistical error from sampling only $T$ orderings of Hamiltonian terms is below $\Theta(\Delta t^{q}N^{1/2}/T)$, where $q$ is the leading-order term for the simulation error of each product formula. This result shows that the weighted average implemented by random sampling can efficiently converge to the expectation regardless of $\Gamma$.

We now consider long-time simulations and numerically illustrate the effectiveness of averaging unitary circuits on the example of the one-dimensional Heisenberg model with power-law interactions:
\begin{eqnarray} 
H&=&H_X+H_Y+H_Z \nonumber\\
&=& \sum_{j=1}^{n-1}\sum_{i=j+1}^n\frac{1}{\abs{i-j}^\alpha}\bigg(X_i X_j+ Y_i Y_j+Z_i Z_j\bigg). \label{eq:PowerHeisenberg}
\end{eqnarray}
Here $H_{X, Y, Z}$ are the Hamiltonian terms containing the products of Pauli matrices $X_i X_j$, $Y_i Y_j$, and $Z_i Z_j$, respectively. The power-law exponent $\alpha$ describes how fast the interaction decays with the distance between the spins. 
We consider the simulation of $\exp(-iHt)$ for $0\leq t \leq 100$.
As mentioned in Sec.~\ref{sec:GenFrame}, we choose one ordering and repeatedly use the Suzuki-Trotter formula of this ordering in each time step. We then average over different orderings with the corresponding weights. We derived the analytical value for the optimal mixing probability $p_{opt}$ in NUSCs only in the case when Hamiltonians consist of two non-commuting terms $H=A+B$ [Eq.~(\ref{eq:2ordBestP})]. In a more general case, when the number of non-commuting terms is $\geq 3$, the analytical calculation becomes tedious, since it relies on the general expression for the third-order Trotter error.  
In order to compute the optimal vector of weight coefficients numerically, we discretize the interval of values of the weight coefficients $p_m \in [0, 1]$ into $10$ bins, so that there are $11$ possible values $p_m\in\{0,0.1, 0.2,...,1\}$.
Next, we perform a grid search on a discrete lattice of $11^M$ points, where $M$ is the number of contributing circuits for averaging. In the Heisenberg model example we are considering, we have $M=6$ corresponding to the number of permutations of three Hamiltonian terms $H_X$, $H_Y$, and $H_Z$.
In Fig.~\ref{fig:STcomblong}(a), we observe that, for strongly long-range models with $\alpha\leq 1$, the weighted average method reduces the error by an order of magnitude. Even for $\alpha \gtrsim 2$, averaging unitary circuits can still result in a significant $40\%$ reduction in the loss. 
In Fig.~\ref{fig:STcomblong}(b,c), we focus on the $6$-qubit power-law Hamiltonian with $\alpha=0$ and $\alpha=1$. We increase the number of steps and observe that NUSCs reduce $L_F$ by an order of magnitude in both short-time and long-time simulations. In Fig.~\ref{fig:STcomblong}(c), we observe that the ratio of loss functions (i.e.\ error reduction) is different in short-time simulation and long-time simulation. This is because the reduction ratio is different for $E_q$ and $\xi_q$ in \cref{eq:ErrDecomp}. It is also interesting that, in the case of all-to-all interactions ($\alpha=0$), the loss function $L_F$ grows strictly linearly with the number of Trotter steps $N$, see Fig.~\ref{fig:STcomblong}(b). We provide an analytical explanation of this behavior in \Cref{section:ErrDecomp}.

The approach for averaging over product formulas with different Hamiltonian term ordering to reduce the simulation error can be used not only in the context of simulation on quantum hardware but also in classical simulation methods such as the time-evolving block-decimation (TEBD) algorithm. In particular,  in \Cref{section:iTEBD}, we apply the averaging technique to the infinite time-evolving block-decimation algorithm~\cite{Vidal2007Classical} to compute the ground state and the corresponding energy of a spin chain. We provide numerical evidence that averaging unitary simulations can improve the convergence rate of the classical algorithm. We believe this method could also be used to improve a wide range of other classical computational methods for real-time and imaginary-time simulation of quantum many-body systems. 

\subsection{Symmetry transformations}\label{sec:Sym}
\begin{figure*}
    \centering
    \includegraphics[width=0.98\textwidth]{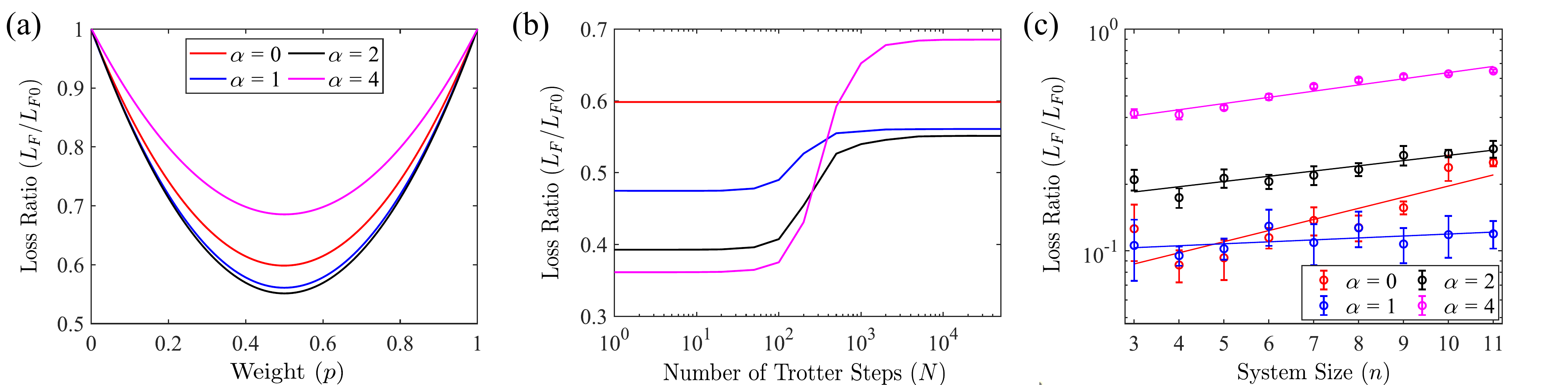}
    \caption{Simulating the power-law Heisenberg model [Eq.~\eqref{eq:PowerHeisenberg}] by averaging unitary circuits using symmetry transformations. Given a formula $U_1$, we construct other formulas in the average by applying symmetry transformations $R(\theta, \phi)$: $U_1 \to R^\dag(\theta, \phi)\, U_1\, R(\theta, \phi)$. (a)~Reduction of the quantum simulation error in the second-order Suzuki-Trotter formula by averaging two unitary circuits constructed from repeatedly using product formulas $\{U_1, \tilde U_1\}$ in each step, where $\tilde U_1 = C^{H\dagger} U_1 C^H$ and $C^H$ is the Hadamard gate applied to all spins, as a function of the weighting coefficient $p$. The Trotterized evolution consists of $N=5\times10^{4}$ time steps with a fixed step size $\Delta t=2\times 10^{-3}$, corresponding to a total simulation time $t=100$.  (b) Same as in (a), but we fix $p = 0.5$ and vary the number $N$ of Trotter steps.
    (c)  We implement $M=10$ different Haar-random symmetry transformations on the second-order Suzuki-Trotter formula $U_1$ for $t=100$ and increase system sizes $3\leq n \leq 11$. We plot the simulation error reduction as a function of system size $n$ for different power-law exponents $\alpha=0,1,2,4$. The solid lines show linear fits on a log-linear plot. The error bar here indicates the standard deviation due to randomly sampled symmetry transformations.}
    \label{fig:SymProb}
\end{figure*}

In this section, we consider the construction of the contributing product formulas (in each step) using symmetry transformations. For that, we assume that the Hamiltonian is invariant under a set of unitary transformations from a group $S$, i.e.~$[C,H] = 0$, for all $C\in S$. 

For the short-time evolution $V(\Delta t) = e^{-iH\Delta t}$ such that $\norm{H}_F\Delta t\ll1$, assume that we have a simulation $U_1(\Delta t)=V(\Delta t)-i\sum_{s=q}^{2q}E_1^{(s)}\Delta t^s+O(\Delta t^{2q+1})$ with leading order error $O(t^q)$, where $E_1^{(s)}$ are the $s$th-order error operators, as defined in Sec.~\ref{sec:GenFrame}. We construct $M$ contributing product formulas by choosing a finite set of unitaries (possibly including the identity operator) $C_0,...,C_{M-1}\in S$ and applying them to the simulation $U_1$: $U_m=C_m^\dagger U_1C_m$. 
Since $[C_m,V]=0$, we can rewrite the weighted combinatorial $q$th-order error operator $E_q$ of the average channel as
\begin{equation}\label{eq:SymProbEq}
E^\text{sym}_q = \sum_{m=0}^{M-1} p_m C_m^\dagger E_1^{(q)}C_m.
\end{equation}
The representation in \cref{eq:SymProbEq} is a weighted average of the error over $M$ simulations each having error $C_m^\dagger E_1^{(q)}C_m$ and weight $p_m$.
We can explain the reduction of error in Eq.~\eqref{eq:SymProbEq} by imagining $E_{1}^{(q)}$ as a vector and $C_m^\dagger E_1^{(q)}C_m$ as a rotation of the original vector. Unlike in Sec.~\ref{sec:TermOrdering}, each contributing simulation is now guaranteed to have the errors of the same norm: $\norm{C_m^\dagger E_1^{(q)}C_m}_F = \norm{E_1^{(q)}}_F$. By combining different simulations, we get a weighted combination of vectors that have different orientations. The length for the resulting vector cannot be longer and is usually smaller compared to the original simulations due to the triangle inequality. 

If we have no information about the error structure, we can choose symmetry transformations $C_m$ randomly from the group $S$. In the original paper on symmetry protection~\cite{Tran2021Faster}, a symmetry transformation is added in each time step, and each time step has the same duration. Compared with that construction, the NUSCs considered here have an additional degree of freedom: we can adjust the weights $\{p_m\}$. Furthermore, in our averaging approach, we use the same symmetry transformation in each time step, so the symmetry transformation $C_m$ in each step will cancel with the $C_m^\dagger$ in the next step, leaving symmetry transformations only at the beginning and at the end of the circuit. Therefore, we no longer need to apply symmetry transformations in each time step, thus potentially reducing the total gate count.

To illustrate error reduction using simulations constructed with symmetry transformations, we carry out numerical simulations on the one-dimensional long-range Heisenberg model in Eq.~\eqref{eq:PowerHeisenberg}. The Hamiltonian is invariant under global rotations $S=\{\bigotimes_{i=1}^n R_i(\theta, \phi):R_i\in SU(2)\}$ with $R_i$ denoting an arbitrary single-qubit rotation $R$ on the $i$-th spin (the same rotation $R$ is applied on each spin). Specifically, we choose $R$ Haar-randomly. Similar to the symmetry protection approach~\cite{Tran2021Faster}, we expect an $O(M^{-\frac12})$ decrease in the simulation error non-commutative with any symmetry $C\in S$ except identity if we assume equal weights $p_m$ and randomly picked symmetries $C_m$. This result can be understood by viewing the rotations $R$ as a random walk in the vector space of operators~\cite{Tran2021Faster}. 

As shown above, randomly chosen symmetry transformations applied to the entire evolution operator can reduce the simulation error. 
We can obtain greater improvement if we know the structure of the simulation error. Here, we consider a special case: $C_{m}=(C_1)^m$ for a fixed unitary operator $C_1$ and $m=0,\ldots,M-1$. We further assume that $C_1=\text{exp}(-i O\Delta)$ is generated by a Hermitian operator $O$ such that $\norm{O}\Delta\ll 1$ and at least one eigenvalue of $O\Delta$ is an irrational multiple of $\pi$. In this case, we can obtain a provable $O(\frac{1}{M})$ reduction of the part of the error that does not commute with $C_1$. Given an arbitrary simulation error $E$, we employ a technique similar to \cref{eq:ErrDecomp} and decompose the error as $E=[O,\eta_C]+\xi_C$, where $[O,\eta_C]$ ($\xi_C$) is the part of the error that does not commute (commutes) with $O$. The derivation of $\xi_C$ and $\eta_C$ is the same as the derivation of $\xi$ and $\eta$, except one needs to replace $H$ with the operator $O$. We then have
\begin{align}\label{eq:SymReduce}
E_q^{\text{sym}}\approx \xi_C+\frac{1}{M\Delta t}\left(e^{i M\Delta tO}\eta_C e^{-iM\Delta t O}-\eta_C\right).
\end{align}
Here, we assume that $M$ is a large number. The derivation of \cref{eq:SymReduce} follows the same pattern as the derivation of \cref{eq:LongTimeErrorRep}. Therefore, by increasing $M$, we can suppress the part of the error that does not commute with $O$. In \Cref{section:Equiv}, we provide the proof of Eq.~\eqref{eq:SymReduce} and benchmark the error reduction achieved with this scheme.

We remark that there are advantages and disadvantages for the two methods. When the number of Hamiltonian terms $\Gamma$ is small, the averaging over term permutations involves only a relatively small number of contributing formulas $M=\Gamma!$ and provides an asymptotic error reduction, whereas averaging over symmetry group $S$ would have $1/M$ convergence rate to the same asymptotic error value as $M\to \infty$. In the opposite limit, when the number of Hamiltonian terms is large and the number of possible term permutations grows exponentially, the averaging over the full permutation group of Hamiltonian terms becomes prohibitive. In such a case, there are two options: (a) perform averaging over a finite number of symmetry transformations $S$,  or (b) use the sampling approach to generate a finite subset of random permutations of the Hamiltonian terms, as described in \cref{eq:2orderErrText}. Depending on the concrete experimental setting, it might be more convenient to use either method (a) or (b) for Hamiltonians with a large number of terms.    

In Fig.~\ref{fig:SymProb}(a,b), we demonstrate the effect of symmetry-based NUSCs for the simulation of the power-law Heisenberg spin chain in Eq.~\eqref{eq:PowerHeisenberg}. We consider $n=6$ qubits with open boundary conditions and apply a single symmetry transformation, resulting in $M=2$ contributing unitaries.  We choose $R_i$ to be $\pi/2$ rotations around the $x$-axis corresponding to the Hadamard gate, which is equivalent to the substitution $X\leftrightarrow Z$ and $Y\leftrightarrow -Y$. Therefore, applying this symmetry transformation to the evolution in this model is equivalent to implementing a product formula with different orderings of the three parts of the Hamiltonian. The symmetry operator can be written as $C^H=(R^H)^{\otimes n}$, where $R^H$ denotes a single-qubit Hadamard gate. In Fig. \ref{fig:SymProb}(a), we observe a $30\%$ to $45\%$ reduction in the average error for different power-law exponents $\alpha$ by using equal weights for the two contributing unitaries. In Fig.~\ref{fig:SymProb}(b), we fix the step size and increase the total simulation time. For different power-law exponents $\alpha$, we observe $40\%$ to $65\%$ reduction in $L_F$ in short-time simulations and $30\%$ to $50\%$ reduction in long-time simulations. Similar to the previous section, we also observe a step-like increase in the error reduction ratio as the number of steps $N$ grows from $N=100$ to $1000$, reflecting the dominance of one of the two terms in \cref{eq:ErrDecomp}. In Fig.~\ref{fig:SymProb}(c), we choose $M = 10$ symmetry transformations randomly from the $SU(2)$ symmetry group. Our method works well for the power-law exponents in the range  $0\leq\alpha\leq 1$  and results in simulation error reduction for different system sizes.

An important question concerning the practical implementation of our approach is whether the error reduction persists for larger systems without changing the choice of weights. In general, it is hard to derive an analytical bound with favorable dependence on system size $n$ for the NUSC error (the loss function) in Eq.~\eqref{eq:GenLoss}. This is because we use the Frobenius norm in order to analytically compute the Haar integral. As the Frobenius norm for spin models can increase exponentially with $n$, it might be hard to get a useful bound in the limit of large $n$. However, we remark that the scaling of error with $n$ will not increase exponentially if we consider the spectral norm in Eq.~\eqref{eq:GenLoss} instead of the Frobenius norm.
While the spectral norm is a common choice in most analytical results on digital quantum simulation, we cannot compute Eq.~\eqref{eq:GenLoss} using the spectral norm. 
We would like to note that, as $n$ increases, the slope for the system-size scaling for the power-law exponents $\alpha=1$ and $\alpha=2$ is small in Fig. \ref{fig:SymProb}. For power-law models, we can thus expect that the error reduction due to mixing product formulas persists up to system sizes $n\sim 100$, where any classical simulation becomes computationally hard.

\section{Experimental Benchmarks on the IonQ Quantum Processor}\label{sec:BraketExp}
\begin{figure}
    \centering
    \includegraphics[width=0.49\textwidth]{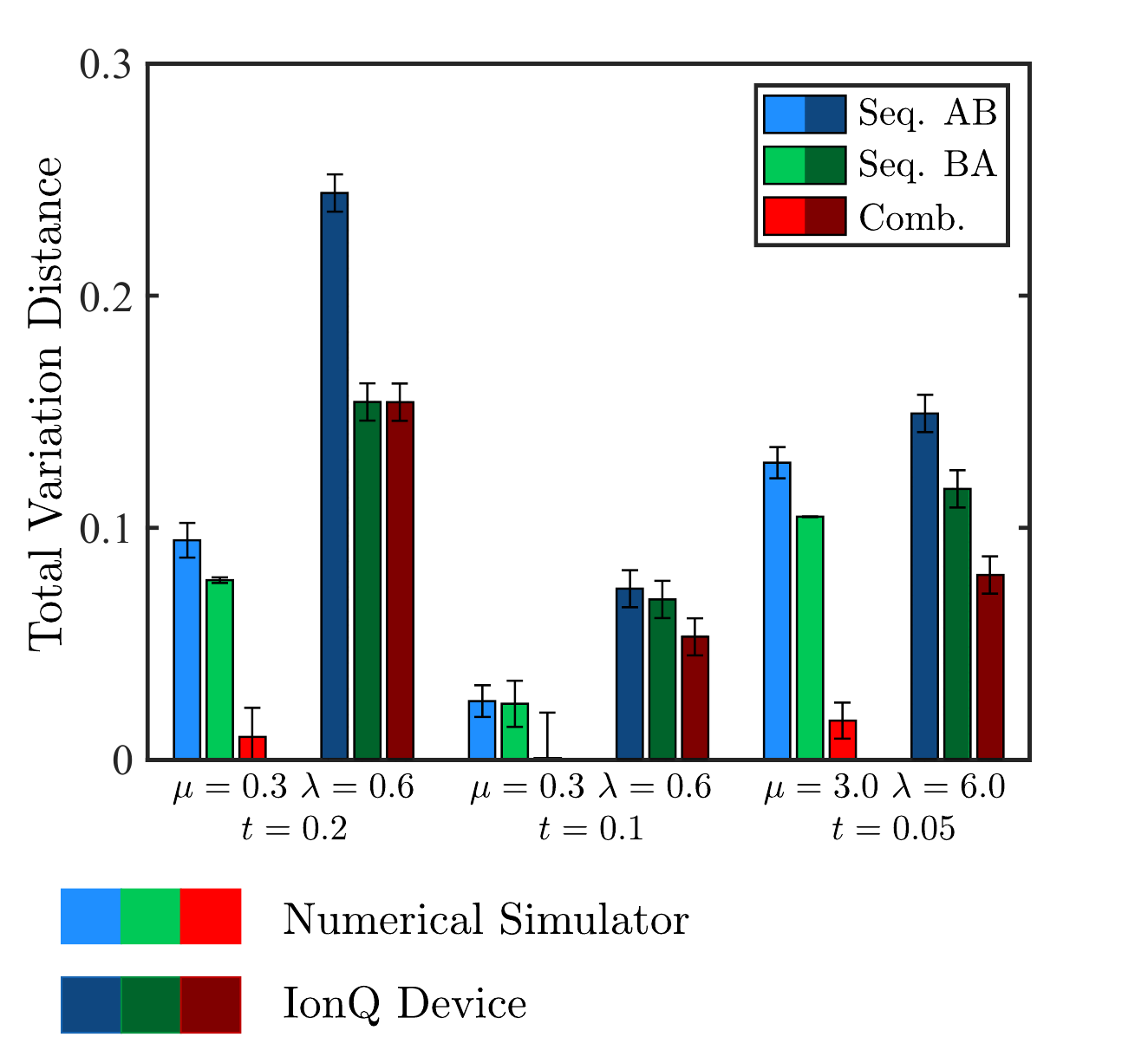}
    \caption{Experimental error suppression using a linear combination of product formulas to simulate time evolution under the Hamiltonian in Eq.~\eqref{eq:IsingTransLong}. We plot the total variation distance between the measured probability distribution of bitstrings and the ideal distribution obtained by solving the Schr\"{o}dinger equation: $\textrm{Total Variation Distance} \!=\! \frac{1}{2}\sum_{u=1}^{2^n} \left|p_u^{(\text{meas})}-p_u^{(\text{Schrod})}\right|$, where $p_u^{(\text{Schrod})}=|\psi_u|^2$ is the squared amplitude of the wavefunction for the computational basis state $u$.  
    For each set of simulation parameters $(\mu,\lambda,t)$, we show two groups of histograms: (left) the ideal numerical simulator and (right) the IonQ Harmony device. To distinguish these two groups, the simulations on the numerical simulator are depicted using lighter colors,  while the experiments on the IonQ quantum processor are depicted using darker colors. Within each histogram group, we compare the performance of individual first-order product formulas, $e^{-iAt}e^{-iBt}$ (blue) and $e^{-iBt}e^{-iAt}$ (green), and their weighted linear combination with equal weights $\vec p=\{0.5,0.5\}$ (red). All simulations on the IonQ device were performed for an $n=5$ qubit subsystem, and the number of measurement shots was set to $N_{\text{shot}}=10^4$.
    }
    \label{fig:IonQExp}
\end{figure}

In the previous sections, we presented both analytical results and numerical benchmarks for averaging product formulas. 
To provide further benchmarks of NUSCs, we carry out experiments on IonQ's 11-qubit trapped-ion quantum computer \cite{Wright2019Benchmarking} via the Amazon Braket cloud platform~\cite{AmazonBraket}.
The IonQ Harmony device features all-to-all connectivity with average single-qubit and two-qubit gate fidelity of $99.5\%$ and $97.5\%$, respectively.

We consider a Hamiltonian for the Ising model in both transverse and longitudinal magnetic fields with periodic boundary conditions~\cite{Atas2017Quantum}:
\begin{align}\label{eq:IsingTransLong}
H=\underbrace{\sum_{i=1}^nX_iX_{i+1}+\mu\sum_{i=1}^nX_i}_{A}+\underbrace{\lambda\sum_{i=1}^nZ_i}_{B}.
\end{align}
This Hamiltonian can be natively simulated using only $R_{XX}(\theta)$ gates and single-qubit rotations chosen from the IonQ device gate set. 
The model (\ref{eq:IsingTransLong}) is not exactly solvable provided  $\mu$ and $\lambda$ are nonzero and provided $\lambda\neq 1$~\cite{Atas2017Quantum}. We simulate this model for system size $n=5$ on the IonQ device. We consider combining first-order product formulas $e^{-iAt}e^{-iBt}$ and  $e^{-iBt}e^{-iAt}$ with weights $\{p,1-p\}$. We focus on short-time simulations using a single step. This allows us to maximally suppress the experimental noise brought by circuit depth and simulation time. We fix the input state to be $(R^H)^{\otimes n}\ket{0}^{\otimes n}=\frac{1}{\sqrt{2^n}}\sum_{i=0}^{2^n-1}\ket{i}$ (where $\ket{i}$ runs over all computational basis states) and estimate the total variation distance between expected and measured probability distributions in the computational basis. We set the number of measurement shots to be $N_{\text{shot}}=10^4$ for each product formula. 

The results of the experiment are shown in Fig.~\ref{fig:IonQExp}(a). We simulate the dynamics of the Hamiltonian in \cref{eq:IsingTransLong} at different values of $(\mu,\lambda,t)$ on both a numerical simulator and the real IonQ device.  The error on the IonQ device is composed of three parts: the analytical simulation error, the sampling noise as we can only approximate the distribution using limited samples, and the experimental noise. On a numerical simulator, there is no experimental noise. 
We fix $p = 0.5$ to eliminate the second-order error $\propto t^2$. When combining two product formulas, we set the number of shots per product formula to be $N'_{\text{shot}}=N_{\text{shot}}/2 = 5\times 10^3$ for a fair comparison between the contributing product formulas and the average. We found a regime of parameters where the sampling errors remain smaller compared to hardware errors and Trotterization errors. 
We set the simulation times to be small enough $t\in [0, 0.2]$, so that the accumulated errors on the quantum hardware remain comparable to the Trotter error of an ideal simulation. 
Increasing the number $N_{\text{shot}}$ of measurement shots reduces the sampling error, however, this becomes problematic when the initial and final states are close to each other so that a prohibitively large number of shots is required to achieve the necessary level of accuracy. 
As the shot number $N_{\text{shot}}$ is unchanged, the sampling error (represented by the error bar) for the sampling procedure remains the same for different $t$. 
To ensure that the sampling error does not overwhelm the Trotter error in such cases, we choose larger $\mu$ and $\lambda$ for $t=0.05$ in order to increase the Trotter error because the parameters that control the Trotter error are the dimensionless products $\mu t$ and $\lambda t$.

When simulating quantum dynamics classically, NUSCs can achieve an improvement of one order of magnitude. On the other hand, on the IonQ Harmony device, the improvement is negligible for $t=0.2$ but becomes noticeable as we decrease the simulation time from $t=0.2$ down to $t=0.05$. At $t=0.05$, we experimentally demonstrate a $30\%$ reduction in the simulation error. Our results indicate that combining product formulas results in the suppression of the effective Trotter error in experimental settings. Although the Trotter error reduction in NUSCs occurs both at short and long evolution times (as demonstrated in Fig.~\ref{fig:SymProb}), the effect becomes experimentally measurable on noisy devices only at short evolution times. 
Finally, in the most interesting regime when the total evolution time is large, digital quantum simulations split the evolution time into smaller time steps and apply a product formula in each time step. Our approach further applies a NUSC over full circuits of multiple time steps. However, hardware errors depend not only on the total simulation time $t$, but also on the number of gates applied throughout the evolution, since every gate incurs coherent control errors which will accumulate. 
Therefore, we would expect that, in a realistic experimental setting, digital quantum simulation using either NUSCs or standard Suzuki-Trotter formulas will fail only when the number of Trotter steps is below a certain critical value, whereas the hardware error will start to dominate when the number of Trotter steps is above this critical value.

\section{Summary and outlook \label{sec:outlook}}
In summary, we consider an approach to reduce the digital quantum simulation error using a weighted average over a few unitary circuits. We analyze two methods for constructing contributing unitary circuits: (i) taking different orderings of Hamiltonian terms in Suzuki-Trotter formulas and (ii) exploiting symmetry transformations. We show that one can potentially achieve the accuracy of a higher-order simulation via simply averaging the outcome of different unitary circuits regardless of the number of simulation steps (i.e.~regardless of the simulation time), which can save considerable quantum resources. We demonstrate the error reduction via NUSCs using experiments on the IonQ device via the Amazon Braket cloud platform. Finally, in \Cref{section:iTEBD},  we demonstrate numerically that our approach for averaging unitary simulations can improve the classical infinite time-evolving block decimation algorithm for simulating quantum systems \cite{Vidal2007Classical}.

While we provide the optimal weight distribution for the lowest-order error only in the case of Hamiltonians made up of two directly realizable terms ($H=A+B$), it would be important to provide an analytical bound for the error reduction for a general Hamiltonian that contains more than $2$ experimentally realizable terms. It would also be interesting to analytically prove the error reduction of NUSCs and characterize how the error reduction scales with system size in long-time simulations.

The figure of merit exploited in this work evaluates the performance of a given simulation by calculating the average simulation error on pure states. Error analysis averaged over all quantum states in the low-energy subspace of a given Hamiltonian, building upon the result in Ref.~\cite{Csahinouglu2021Hamiltonian}, would be another interesting future direction.

It is also important and interesting to generalize the analysis in this paper to more advanced quantum simulation methods such as truncated Taylor series \cite{BerryCCKS2015} or qubitization \cite{Low2019Hamiltonian}, as the error structure would be typically more complicated. 

Finally, we emphasize that the optimal weights in our averaging approach depend on the exact error structure determined by both the simulation algorithm and the properties of the Hamiltonian. It would be interesting to calculate the exact optimal weights for the symmetry transformation approach and the permutation-of-terms approach for some physically relevant models.

\section{Acknowledgements}
We thank Tongyang Li and Dong-Ling Deng for helpful discussions. Y.K., M.C.T., P.B., and A.V.G.~were supported in part by NSF QLCI (award No.~OMA-2120757), DoE ASCR Accelerated Research in Quantum Computing program (award No.~DE-SC0020312), NSF PFCQC program, the DoE ASCR Quantum Testbed Pathfinder program (award No.~DE-SC0019040), AFOSR, ARO MURI, AFOSR MURI, and DARPA SAVaNT ADVENT. Support is also acknowledged from the U.S.~Department of Energy, Office of Science, National Quantum Information Science Research Centers, and Quantum Systems Accelerator. 

\bibliographystyle{apsrev4-1-title}
\bibliography{papers,product-formula}

\onecolumngrid
\newpage
\appendix
\makeatletter

\begin{center} 
    {\large \bf Appendix for: Improved Digital Quantum Simulation by Non-unitary Channels}
\end{center} 

In this Appendix, we provide additional details on numerical calculations in the main text and expand upon the theoretical aspects of this work. 
In \Cref{section:ProofThm1}, we provide a detailed proof of \Cref{thm:average-error} in the main text. 
In \Cref{section:ProofThm2}, we prove \Cref{thm:opt-weight} in the main text and provide more numerical evidence. 
In \Cref{section:ProofInduct}, we provide the proof of \cref{eq:STHighAppr} in the main text. 
In \Cref{section:ErrDecomp}, we decompose the simulation error into two terms and study the behavior of each term in short-time simulations and long-time simulations. We also explain the linear growth of simulation error with the number of Trotter steps in \Cref{fig:STcomblong}(b) at $\alpha=0$. Finally, we give the proof of Lemma~\ref{clm:ShortABErrRed} in the main text.
In \Cref{section:iTEBD}, we use NUSCs in the classical infinite time-evolving block decimation (iTEBD) algorithm to calculate the ground state and its energy for one-dimensional quantum models. We numerically show that NUSCs could considerably accelerate the convergence rate of the algorithm. 
In \Cref{section:ConvRandSam}, we provide the step and sample complexity calculations to bound the random fluctuations when implementing random sampling for Hamiltonians with a large number of terms in \Cref{sec:TermOrdering}.
In \Cref{section:Equiv},  we provide additional numerical results for the symmetry protection approach in \Cref{sec:Sym}. We also provide analytical and numerical verification for the $O(1/M)$ reduction for combining a simulation under a particular group of symmetry protection considered in \Cref{sec:Sym}: $\{C_{m}=(C_1)^m\}_{m=0,\ldots,M-1}$. 

\section{Proof of \Cref{thm:average-error}}
\label{section:ProofThm1}
In this section, we prove \Cref{thm:average-error} in the main text.

Suppose that we have $M$ contributing simulations $U_1,...,U_M$ and the target evolution $V$. The average error $L_F$ in \cref{eq:GenLoss} for the NUSC can be written as
\begin{align}\label{eq:Ldef}
& L_F(U_1,...,U_M,p_1,...,p_M,V)=\int d|\psi\rangle
\norm{\sum_{m=1}^Mp_mU_m\ketbra{\psi}U_m^\dagger-V\ketbra{\psi}V^\dagger}_F^2,
\end{align}
where $\norm{\cdot}_F$ is the Frobenius norm, and the integral is over pure states drawn from the Haar measure. When $M = 1$, this integral can be evaluated analytically~\cite{Ratcliffe1994Foundations}:
\begin{align}\label{eq:Haar-loss-F}
&L_F(U,V)=\int d|\psi\rangle\norm{U\ketbra{\psi}U^\dagger-V\ketbra{\psi}V^\dagger}_F^2
=2\left[1-\frac{1}{d(d+1)}\left(d+\abs{\Tr\left(U^\dagger V\right)}^2\right)\right],
\end{align}
where $d$ is the dimension of the system. This result makes use of the observation that $U^\dagger\ketbra{\psi}U$ and $V^\dagger\ketbra{\psi}V$ are pure states and of the calculations of Haar integrals in Refs.~\cite{Poland2020No,Zhang2014Matrix,Nielsen2002Simple}.

The choice of using the Frobenius norm enables analytical evaluation of the average error for any $M$:
\begin{align}\label{eq:A7}
L_F(U_1,...,U_M,p_1,...,p_M,V)=&2\left(1-\frac{1}{d+1}\right)\left(\sum_{m=1}^Mp_m^2+\sum_{m>n}p_mp_n\right)\nonumber\\
&-\frac{2}{d(d+1)}\left[\sum_{m=1}^Mp_m\abs{\Tr(U_m^\dagger V)}^2-\sum_{m>n}p_mp_n\abs{\Tr(U_m^\dagger U_n)}^2\right].
\end{align}

Now we provide the proof of \cref{thm:average-error} in the main text. 
Given a product formula with leading order error $O(t^q)$, we expand the approximation error up to the $(2q+1)$-th order. Since each simulation $U_m=V-i\sum_{s=q}^{2q+1}E_{m}^{(s)}t^s+O(t^{2q+2})$ is unitary, i.e. $U_m^\dag U_m = V^\dag V = \mathbb I$, the error operators $E_{m}^{(s)}$ satisfy
\begin{align}\label{eq:A10}
i\sum_{s=q}^{2q+1}t^s\left(E_m^{(s)\dagger} V-V^\dagger E_m^{(s)}\right)+t^{2q}E_m^{(q)\dagger} E_m^{(q)}+t^{2q+1}\left(E_m^{(q+1)\dagger} E_m^{(q)}+E_m^{(q)\dagger} E_m^{(q+1)}\right)=O\left(t^{2q+2}\right).
\end{align}
We exploit this constraint to expand $\Tr(U_m^\dagger V)$ and $\Tr(U_m^\dagger U_n)$. Ignoring higher order terms $O(t^{2q+2})$, we have
\begin{align}
&\Tr\left(U_m^\dagger V\right)=d+i\sum_{s=q}^{2q+1}t^s\Tr\left(E_{m}^{(s)\dagger} V\right),\\
&\abs{\Tr(U_m^\dagger V)}^2=d^2-t^{2q}\left[d\Tr(E_{m}^{(q)\dagger} E_{m}^{(q)})-\abs{\Tr(E_{m}^{(q)\dagger} V)}^2\right]\nonumber\\
&\quad\quad\quad\quad\quad\quad-2t^{2q+1}\Re\left[d\Tr(E_m^{(q+1)\dagger} E_m^{(q)})-\Tr(E_{m}^{(q+1)\dagger} V)\Tr(V^\dagger E_{m}^{(q)})\right],\\
&\Tr(U_m^\dagger U_n)=d+i\sum_{s=q}^{2q+1}t^s\left[\Tr(E_{m}^{(s)\dagger} V)-\Tr(V^\dagger E_{n}^{(s)})\right]+t^{2q}\Tr(E_{m}^{(q)\dagger} E_{n}^{(q)})\nonumber\\
&\quad\quad\quad\quad\quad\quad-t^{2q+1}\left[\Tr(E_{m}^{(q+1)\dagger} E_{n}^{(q)})+\Tr(E_{m}^{(q)\dagger} E_{n}^{(q+1)})\right],\\
&\abs{\Tr(U_m^\dagger U_n)}^2=d^2-t^{2q}\left\{d\left[\Tr(E_{m}^{(q)\dagger} E_{m}^{(q)})+\Tr(E_{n}^{(q)\dagger} E_{n}^{(q)})-2\text{Re}(\Tr(E_{m}^{(q)\dagger} E_{n}^{(q)})\right]-\abs{\Tr(E_{m}^{(q)\dagger} V)}^2\right.\nonumber\\
&\qquad\qquad\left.-\abs{\Tr(E_{n}^{(q)\dagger} V)}^2+2\Tr(E_{m}^{(q)\dagger} V)\Tr(E_{n}^{(q)\dagger} V)\right\}-2t^{2q+1}\left\{d\Re\left[\Tr(E_m^{(q+1)\dagger} E_m^{(q)})+\Tr(E_n^{(q+1)\dagger} E_n^{(q)})\right]\nonumber\right.\\
&\qquad\qquad+d\left[\Tr\left(E_{m}^{(q+1)\dagger} E_{n}^{(q)}\right)+\Tr\left(E_{m}^{(q)\dagger} E_{n}^{(q+1)}\right)\right]-\Re\left[\Tr(E_{m}^{(q+1)\dagger} V)\Tr(V^\dagger E_{m}^{(q)})-\Tr(E_{n}^{(q+1)\dagger} V)\Tr(V^\dagger E_{n}^{(q)})\right.\nonumber\\
&\qquad\qquad\left.\left.-\Tr(E_{n}^{(q+1)\dagger} V)\Tr(V^\dagger E_{m}^{(q)})-\Tr(E_{m}^{(q+1)\dagger} V)\Tr(V^\dagger E_{n}^{(q)})\right]\right\}.
\end{align}
Hence, we obtain the lowest two orders of the average error $L_F$:
\begin{align}\label{eq:A12}
L_F=&\frac{2t^{2q}}{d(d+1)}\left[d\Tr(E_q^\dagger E_q)-\abs{\Tr(E_q^\dagger V)}^2\right]+\frac{4t^{2q+1}}{d(d+1)}\Re\left[d\Tr(E_{q+1}^\dagger E_q)-\Tr(E_{q+1}^\dagger V)\Tr(V^\dagger E_q)\right],
\end{align}
where $E_q=\sum_{m=1}^M p_mE_{m}^{(q)}$ and $ E_{q+1}=\sum_{m=1}^Mp_mE_{m}^{(q+1)}$. 
\cref{eq:A12} is exactly the conclusion of \cref{thm:average-error}. 
We note that the above equation also guarantees that the $(2q)$-th order error is always positive. 
This is because $d\Tr(A^\dagger A)\geq{\abs{\Tr(A)}^2}$, where $A=E_q^\dagger V$.

\section{Proof of \Cref{thm:opt-weight}}\label{section:ProofThm2}
\begin{figure}
    \centering
    \includegraphics[width=0.49\textwidth]{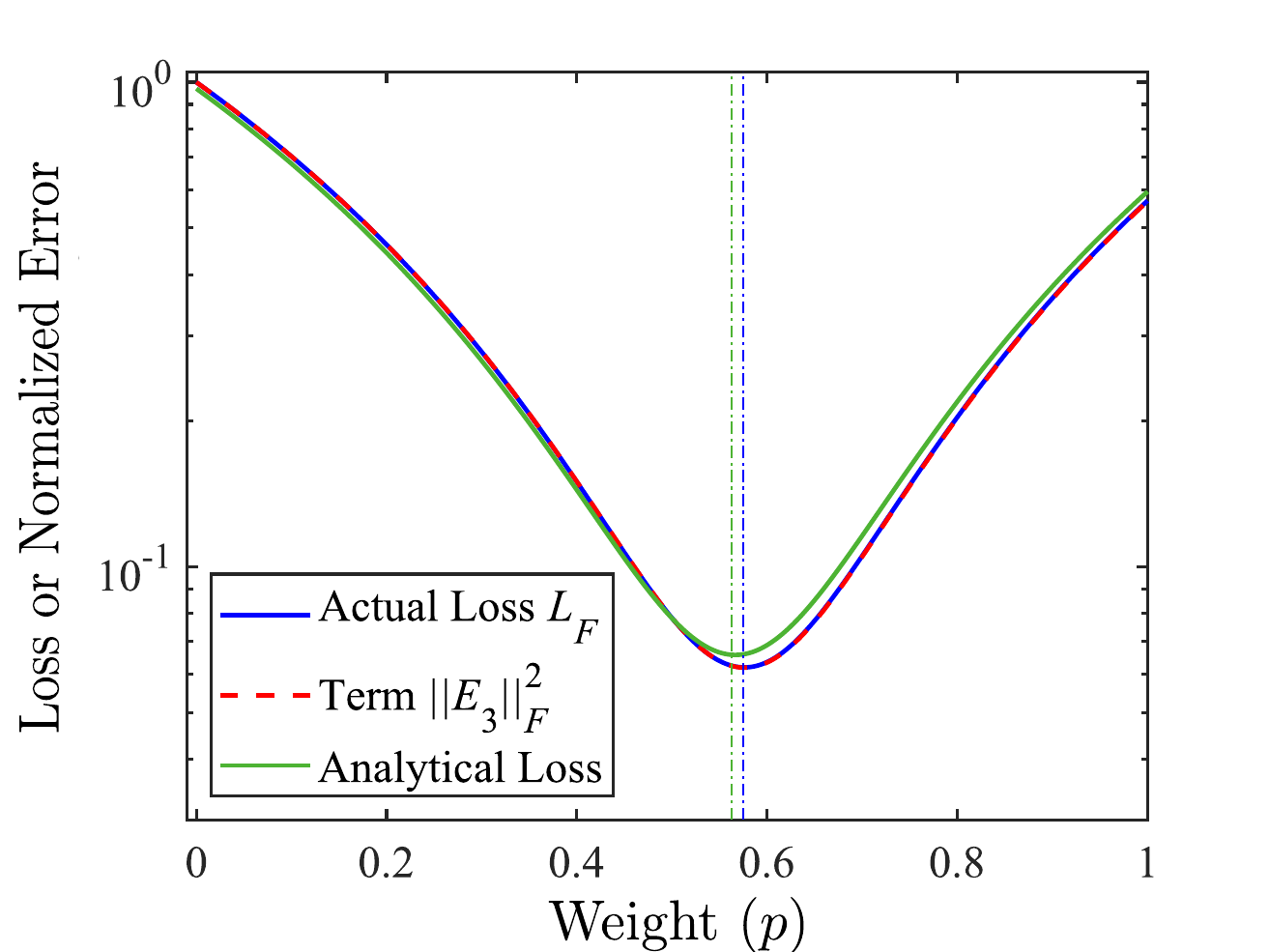}
    \caption{Illustrative numerical simulation of the one-dimensional spin model $H=\sum_{i}X_iX_{i+1}+JX_i+\sum_jY_jY_{j+1}$. We plot the actual average error (i.e.\ loss) $L_F$ (blue), the numerical result for $\norm{E_3}_F^2$ (red), and the analytical result for $\norm{E_3}_F^2$ (green). All curves are normalized to be $1.0$ at $p = 0$ (where we only apply $U_1$). 
The optimal weight for the (approximate) analytical result is $p_{\text{opt}}=0.5625$, while the actual optimal weight is very close: $p_{\text{opt}}=0.5750$.}
    \label{fig:Thm2App}
\end{figure}
In this section, we present a detailed proof of \Cref{thm:opt-weight} in the main text and provide additional numerical evidence illustrating the Lemma. 

As defined in Eqs.~(\ref{eq:suzukiini1},\ref{eq:suzukiini2}) in the main text, $U_1$ and $U_2$ are the two possible second-order Trotterizations of $A$ and $B$, which are the two experimentally realizable parts of the Hamiltonian $H = A + B$. Following the notation in \Cref{section:ProofThm1}, we derive the third-order error for the two product formulas and the NUSC according to Ref.~\cite{Childs2021Theory}: 
\begin{align}\label{eq:D1}
& E_{1}^{(3)}=\frac{1}{12}\left[B,\left[B,A\right]\right]-\frac{1}{24}\left[A,\left[A,B\right]\right],\\\label{eq:D2}
& E_{2}^{(3)}=\frac{1}{12}\left[A,\left[A,B\right]\right]-\frac{1}{24}\left[B,\left[B,A\right]\right],\\\label{eq:D3}
& E_{3}=\frac{3p-1}{24}\left[B,\left[B,A\right]\right]+\frac{2-3p}{24}\left[A,\left[A,B\right]\right]=\frac{1}{24}\left[(2-3p)A+(1-3p)B,\left[A,B\right]\right].
\end{align}

By exploiting the triangle inequality and the sub-multiplicative property of the norm, we can bound the norm of $E_3$ by
\begin{equation}\label{eq:D4}
\norm{E_{3}}_F\leq\frac{1}{12}\norm{(2-3p)A+(1-3p)B}_F\norm{\left[A,B\right]}.
\end{equation}

This bound indicates that the norm of the averaged third-order error 
depends on 
$\norm{(2-3p)A+(1-3p)B}_F$. We write the Hamiltonian as $H = \sum_{\bm \sigma} J_{\bm \sigma} \bm \sigma$, where $\sigma$ are Pauli strings, and assume that $A$ and $B$ have support on non-overlapping sets $S$ and $S^c$ of Pauli strings. In this case, we have $\Tr\left[\left(\alpha A+\beta B\right)^2\right]=\alpha^2\Tr(A^2)+\beta^2\Tr(B^2)$, so that 
\begin{equation}\label{eq:D5}
\norm{E_{3}}_F^2 \leq\frac{1}{12} 2^n \left[  (2-3p)^2 J_A^2 + (1-3p)^2 J_B^2\right]\norm{\left[A,B\right]}^2,
\end{equation}
where $J_A^2=\sum_{\bm \sigma \in S}J_{\bm \sigma}^2$, $J_B^2=\sum_{\bm\sigma\in S^c}J_{\bm \sigma}^2$, and $n$ is the number of qubits. \Cref{eq:D5} also allows us to derive the optimal weight $p_{opt}=\frac{2J_A^2+J_B^2}{3(J_A^2+J_B^2)}$ and the corresponding improvement in \Cref{thm:opt-weight} in the main text.

To further illustrate the result of \Cref{thm:opt-weight}, we provide the numerical simulation of the one-dimensional 
spin model $H=\sum_{i}X_iX_{i+1}+JX_i+\sum_jY_jY_{j+1}$. As shown in Fig.~\ref{fig:Thm2App}, we consider averaging two second-order Trotterizations given in Eqs.~(\ref{eq:suzukiini1},\ref{eq:suzukiini2}) in the main text with $A = \sum_{i}X_iX_{i+1}+JX_i$ and $B = \sum_jY_jY_{j+1}$. We set the system size to be $n=6$, the simulation time to be $t=0.1$, and the parameter $J$ to be $J=1$. In this case, we have $J_A^2= 2^n(2n-1)$ and $J_B^2=2^n(n-1)$. By \Cref{thm:opt-weight}, we can obtain the optimal weight to be $p_{\text{opt}}=0.5625$ and the error to be $0.070$ and $0.111$ of the errors of $U_1$ and $U_2$, respectively. Here, the absolute value of the errors for $U_1$, $U_2$, and the average with the optimal weights are $0.0213$, $0.0130$, and $0.0023$. We observe from the numerical results in  Fig.~\ref{fig:Thm2App} that the actual loss $L_F$ and error $\norm{E_3}_F^2$ deviate from the analytical result slightly due to the higher-order error.

\section{Proof of Equation \eqref{eq:STHighAppr}}\label{section:ProofInduct}
In this section, we provide the proof of \cref{eq:STHighAppr} in the main text. 

As mentioned in the main text, the recursive definition for higher-order Suzuki-Trotter product formulas is given by
\begin{align}\label{eq:E1}
&S_2(t) = \prod_{j=1}^\Gamma\text{exp}\left(iH_j\frac{t}{2}\right)\prod_{j=\Gamma}^1\text{exp}\left(iH_j\frac{t}{2}\right),\\\label{eq:E2}
&S_{k}(t) =  S^2_{k-2}\left(u_kt\right)S_{k-2}\left(\left(1-4 u_k\right)t\right)S_{k-2}^2\left(u_kt\right),
\end{align}
where $u_k=\left(4-4^{1/(k-1)}\right)^{-1}$. Here, $S_k(t)$ provides a digital quantum simulation with $(k+1)$th-order error in $t$. Without loss of generality, we can assume that $E_{S_k}(t)$ has the form
\begin{equation}\label{eq:E3}
E_{S_k}(t)=K^k_{k+1}t^{k+1}+K_{k+2}^kt^{k+2}+K_{k+3}^kt^{k+3}+O(t^{k+4}),    
\end{equation}
where $K^k_{i}$ denotes the error coefficient matrix of th $i$-th order error for $E_{S_k}(t)$. By the recursive relation given by Eq. \eqref{eq:E1}, we can derive the following relation between $E_{S_k}(t)$ and $E_{S_{k-2}}(t)$ under the assumption $\norm{H}_Ft\ll1$:
\begin{equation}\label{eq:E4}
\begin{aligned}
E_{S_k}(t)=&E_{S_{k-2}}(u_kt)e^{iH(1-u_k)t}+e^{iHu_kt}E_{S_{k-2}}(u_kt)e^{iH(1-2u_k)t}+e^{iH(1-u_k)t}E_{S_{k-2}}(u_kt)\\
&+e^{iH(1-2u_k)t}E_{S_{k-2}}(u_kt)e^{iHu_kt}+e^{iH2u_kt}E_{S_{k-2}}((1-4u_k)t)e^{iH2u_kt},
\end{aligned}
\end{equation}
where the higher-order terms of the form $E_{S_{k-2}}(u_kt)E_{S_{k-2}}(u_kt)$ and $E_{S_{k-2}}(u_kt)E_{S_{k-2}}((1-4u_k)t)$  are ignored. $E_{S_k}(t)$ has the lowest-order error term of $(k+1)$th-order, while $E_{S_{k-2}}(t)$ has the lowest-order error term of $(k-1)$th-order. We consider the $k$th-order error on both the left-hand and the right-hand side of Eq.~\eqref{eq:E4}. The coefficient on the left-hand side is $0$, while the coefficient on the right-hand side can be computed to yield the following equation:
\begin{equation}\label{eq:E5}
\begin{aligned}
0=2u_k^{k-1}(5u_k-1)t^k\cdot\left(2K_k^{k-2}-i\left\{K_{k-1}^{k-2},H\right\}\right),\\
\end{aligned}
\end{equation}
where $\{A,B\}=AB+BA$ is the anticommutator. We utilized the fact that $4u_k^{k-1}+(1-4u_k)^{k-1}=0$. 

Now we consider the difference between the left-hand side and the right-hand side of \cref{eq:STHighAppr} in the main text. Taking advantage of the representation in Eq.~\eqref{eq:E4}, this difference can be written as
\begin{equation}\label{eq:E6}
2e^{iHu_kt}E_{S_{k-2}}(u_kt)e^{iH(1-2u_k)t}+2e^{iH(1-2u_k)t}E_{S_{k-2}}(u_kt)e^{iHu_kt}+e^{iH2u_kt}E_{S_{k-2}}((1-4u_k)t)e^{iH2u_kt}.
\end{equation}

The $k$th-order term in \cref{eq:E6} gives exactly the right-hand side of Eq.~\eqref{eq:E5} and hence vanishes. We now consider the $(k+1)$th-order term in Eq.~\eqref{eq:E6}. It gives
\begin{equation}\label{eq:E7}
\begin{aligned}
&-4K_{k+1}^{k-2}u_k^{k-1}t^{k+1}(5u_k-1)(3u_k-1)\\
&+u_k^{k-1}t^{k+1}\left[i\left\{K_k^{k-2},H\right\}\cdot 6u_k(5u_k-1)+\left(K_k^{k-2}H^2+H^2K_k^{k-2}\right)\cdot\left(3u_k^2+4u_k-1\right)+HK_k^{k-2}H\cdot(24u_k^2-4u_k)\right].
\end{aligned}
\end{equation}
Given that, as $k\to\infty$, $u_k\to1/3$, we have the first term approaching $0$ and the rest of the terms approaching to $2/3u_k^{k-1}t^{k+1}\left\{(2K_k^{k-2}-i\left\{K_{k-1}^{k-2},H\right\}),iH\right\}$ in this limit. Then, using Eq.~\eqref{eq:E5}, we find these remaining terms also vanish. Therefore, the leading-order error of \cref{eq:E6} has a scaling of $O(t^{k+2})$ when $k\to\infty$. According to \cref{eq:E4}, the difference between the left-hand side and the right-hand side of \cref{eq:STHighAppr} is exactly \cref{eq:E6} and has a $(k+2)$th order leading-order term. This concludes the proof of \cref{eq:STHighAppr} in the main text.

In addition, in Table \ref{tab:Thm4Finite}, we provide numerical values for some of the parameters in  Eq.~\eqref{eq:E7} and their ratios for $k=4,6,8,10, \infty$.
\begin{table}[H]
    \centering
    \begin{tabular}{|r|r|r|r|}
        \hline
        $k$ & $6u_k(5u_k-1)$ & $(3u_k^2+4u_k-1)/6u_k(5u_k-1)$ & $(24u_k^2-4u_k)/6u_k(5u_k-1)$\\
        \hline
        $4$ & $2.667$ & $0.440$ & $0.924$\\
        \hline
        $6$ & $1.937$ & $0.470$ & $0.954$\\
        \hline
        $8$ & $1.722$ & $0.480$ & $0.967$\\
        \hline
        $10$ & $1.619$ & $0.485$ & $0.974$\\
        \hline
        $\infty$ & $\frac{4}{3}$ & $\frac12$ & $1$\\
        \hline 
    \end{tabular}
    \caption{The parameters in Eq.~\eqref{eq:E7} for some finite $k$ and for infinite $k$}
    \label{tab:Thm4Finite}
\end{table}

From Table \ref{tab:Thm4Finite}, we see that, when $k=8$, the ratios $(3u_k^2+4u_k-1)/6u_k(5u_k-1)$ and $(24u_k^2-4u_k)/6u_k(5u_k-1)$  deviate by only about $5\%$ from their ideal values at $k\to\infty$. Since the average error $L_F$ is defined using the square of the norm, such a deviation would only cause an $\approx 10^{-3}$ deviation.

\section{Error decomposition for long-time simulations}\label{section:ErrDecomp}
In this section, we consider the error structure of a $N$-step simulation $U(N,\Delta t)$. Suppose the exact evolution in a time step $\Delta t$ is $V(\Delta t)=e^{-iH \Delta t}$. We consider a simulation circuit $U(\Delta t)=V(\Delta t)-E(\Delta t)$, where $E(\Delta t)$ is the simulation error. We can decompose $E(\Delta t)$ into
\begin{align}\label{eq:B1}
E(\Delta t)=\left[H, \eta(\Delta t)\right]+\xi(\Delta t)  
\end{align}
for some operators $\eta$ and $\xi$ such that $\xi$ commutes with $H$.
In terms of $\eta$ and $\xi$, the error for the $N$-step simulation can be approximated to leading order as
\begin{align}
E(t)&\approx\sum_{j=1}^N e^{-iH(N-j)\Delta t}E(\Delta t)e^{iH(N-j)\Delta t}
\approx\int_{0}^te^{-iHt'}\frac{E(\Delta t)}{\Delta t}e^{iHt}dt'\\
&=\xi(\Delta t)N+\frac{1}{\Delta t}\left[e^{-iHt}\eta(\Delta t)e^{iHt}-\eta(\Delta t)\right].
\label{eq:long-time-error}
\end{align}
Notice that, at a fixed $\Delta t$, the second term in \cref{eq:long-time-error} is bounded by $\frac{2\norm{\eta}}{\Delta t}$, independent of $N$. Meanwhile, the first term increases linearly with $N$ and hence becomes the dominant error in the long-time (large-$N$) limit. 

To derive an analytical form of $\eta$ and $\xi$, we consider an $n$-qubit Hamiltonian $H$ represented by its eigenvalues and eigenstates as
\begin{align}\label{eq:B5}
H=\sum_{i=1}^d E_i|i\rangle\langle i|,    
\end{align}
where $d=2^n$ is the dimension of the Hilbert space, $\{|i\rangle\}$ are the eigenvectors, and $\{E_i\}$ are the energy spectrum  of the Hamiltonian. For simplicity, we suppose the error $E(\Delta t)=E=\sum_{i,j}e_{ij}|i\rangle\langle j|$. We can calculate $\xi$ as
\begin{align}\label{eq:B6}
\xi&=\lim_{T\to\infty}\frac{1}{T}\int_{0}^Te^{-iHt}E e^{iHt}dt\\\label{eq:B7}
&=\sum_{i,j:E_i=E_j}e^{i(E_j-E_i)t}|i\rangle\langle j|e_{ij}.
\end{align}

Now we consider an example of the all-to-all-interacting Heisenberg model: 
\begin{equation}\label{eq:H_all2all_Heisenberg}
    H = \sum_{i<j} \left(X_i X_j + Y_i Y_j + Z_i Z_j \right)
\end{equation}
and analyze the Trotter error in simulating time evolution under $H$. This Hamiltonian corresponds to the long-range-interacting Heisenberg model with the power-law exponent $\alpha=0$. 
The presence of additional symmetries in \cref{eq:H_all2all_Heisenberg} allows us to compute the error terms analytically.
We consider this special model in order to explain why, in \Cref{fig:STcomblong}(b) and \Cref{fig:SymProb}(b) in the main text,  the $\alpha = 0$ model does not undergo a transition in what error term dominates \cref{eq:long-time-error}.

We can simplify our calculations by noticing that the Hamiltonian (\ref{eq:H_all2all_Heisenberg}) can be written in terms of a large spin:
\begin{equation}
    H = 2 ( S_x^2 + S_y^2 + S_z^2 ) - \frac{3}{4} N = H_x+H_y+H_z  - \frac{3}{4} N,
\end{equation}
where $H_{x,y,z}=2 S_{x,y,z}^2$ and the components of the "large spin" are defined sums of Pauli operators: $S_x=\frac{1}{2}\sum_{i} X_i$, $S_y=\frac{1}{2}\sum_{i} Y_i$, $S_z=\frac{1}{2}\sum_{i} Z_i$.
Components of the spin operator satisfy standard commutation relations:
\begin{equation}
    [S_\alpha, S_\beta] = i\epsilon_{\alpha\beta\gamma}S_\gamma.
\end{equation}

To construct the Trotter product formula, we split the Hamiltonian into three terms $H_x$, $H_y$, and $H_z$. 
We first consider the error term for the first-order product formula.
The error has the following structure: 
\begin{equation}
    \mathcal{E}_{q=2} \propto \Delta t^2 \sum_{\mu>\nu} [H_\mu, H_\nu].
\end{equation}
We will show that the second-order error $\mathcal{E}_{q=2}$ commutes with $H$.
Due to the rotational symmetry of the Hamiltonian, we only need to show that $[H, C_1]=0$, where 
\begin{align}\label{eq:xyz}
    &C_1 = [S_x^2, S_y^2] = S_x[S_x, S_y]S_y + [S_x, S_y]S_x S_y + S_y S_x[S_x,S_y] + S_y[S_x,S_y] S_x \nonumber \\
    &=i\left\{4S_x S_z S_y - 2i(S_y^2+S_z^2) + 2i S_x^2\right\} = i\left(4 S_x S_y S_z - 2iH \right).
\end{align}
From Eq. (\ref{eq:xyz}) it follows that we only need to compute $[H, S_x S_y S_z]$. We now separately calculate the contribution to this commutator from each of the three terms in $H$. The contribution from $S_x^2$ reads
\begin{align}
    &[S_x^2, S_x S_y S_z] = S_x [S_x^2, S_y] S_z + S_x S_y [S_x^2, S_z] = S_x \left(S_x[S_x, S_y]S_z + [S_x, S_y] S_x S_z\right) +\nonumber\\
    &S_x S_y \left(S_x[S_x, S_z] + [S_x, S_z]S_x\right) = i\left(S_x^2 S_z^2 + S_x S_z S_x S_z - S_x S_y S_x S_y - S_x S_y^2 S_x \right).
\end{align}
Next, the contribution from $S_y^2$ reads
\begin{align}
    &[S_y^2, S_x S_y S_z] = [S_y^2, S_x] S_y S_z + S_x S_y [S_y^2, S_z] = i\left(-S_y S_z S_y S_z - S_z S_y^2 S_z + S_x S_y^2 S_x + S_x S_y S_x S_y\right).
\end{align}
Finally, the contribution from $S_z^2$ reads
\begin{equation}
    [S_z^2, S_x S_y S_z] = [S_z^2, S_x] S_y S_z + S_x [S_z^2, S_y] S_z = i\left(S_z S_y^2 S_z + S_y S_z S_y S_z - S_x S_z S_x S_z - S_x^2 S_z^2\right).
\end{equation}
Summing up the three contributions we obtain
\begin{equation}\label{eq:H_C1}
    [H, C_1] = [S_x^2+S_y^2+S_z^2, 4i S_x S_y S_z] = 0.
\end{equation}
\Cref{eq:H_C1} implies that the total error for the $N$-step first-order Suzuki-Trotter formula scales linearly with the number of time steps: $\norm{\mathcal{E}_{q=2}(t)} = N \norm{\mathcal{E}(\Delta t)}$.
Similarly, one can show that the second-order error term $\mathcal{E}_{q=3}$ for the second-order product formula \cite{Childs2021Theory},
\begin{align}\label{eq:2orderErr}
\frac{t^3}{24}\sum_{\gamma_1=1}^\Gamma\left[H_{\gamma_1}+2\sum_{\gamma_2=\gamma_1+1}^\Gamma H_{\gamma_2},\left[\sum_{\gamma_3=\gamma_1+1}^\Gamma H_{\gamma_3},H_{\gamma_1}\right]\right],
\end{align}
commutes with the Hamiltonian (\ref{eq:H_all2all_Heisenberg}).
The commutativity of the second-order error term with the all-to-all-interacting spin Hamiltonian (\ref{eq:H_all2all_Heisenberg}) explains the strictly linear growth of the simulation error with the number of Trotter steps in Fig.~\ref{fig:STcomblong}(b) at $\alpha=0$.

We now give the proof of Lemma~\ref{clm:ShortABErrRed} in the main text. We consider a \textit{non-degenerate} Hamiltonian, i.e.~$E_i\neq E_j$ for $i\neq j$. In this case, for a given simulation error $E$, $\xi$ is just the diagonal part of $E$ in the Hamiltonian eigenbasis:
\begin{align}\label{eq:B8}
\xi=\sum_{i=1}^d e_{ii}|i\rangle\langle i|.
\end{align}

To obtain the expression of $\eta$ satisfying $[H,\eta]=E-\xi$, we write $\eta = \sum_{i \neq j}e^{(1)}_{ij}|i\rangle\langle j|$ and find
\begin{align}\label{eq:B9}
e^1_{ij}=\frac{e_{ij}}{E_i-E_j}.
\end{align}

We now consider the weighted average of second-order Trotterizations for non-degenerate bipartite Hamiltonian $H=A+B$. We write the Hamiltonian as $H=\sum_{i=1}^d E_i|i\rangle\langle i|$ and $A, B$ as
\begin{align}\label{eq:B10}
A&=\sum_{i}\alpha_i|i\rangle\langle i|+\sum_{i\neq j}\alpha_{ij}|i\rangle\langle j|,\\\label{eq:B11}
B&=\sum_{i}(E_i-\alpha_i)|i\rangle\langle i|-\sum_{i\neq j}\alpha_{ij}|i\rangle\langle j|.   
\end{align}

Consider two second-order Suzuki-Trotter formulas $U_1$ and $U_2$ given in Eqs.~(\ref{eq:suzukiini1},\ref{eq:suzukiini2}). Referring to the second term of \cref{eq:B1} for $U_1$ and $U_2$ as $\xi_1$ and $\xi_2$, we find
\begin{align}\label{eq:B12}
\xi_1=\xi_2&\propto\sum_i\kappa_{i}|i\rangle\langle i|,\\ \label{eq:B13}
\kappa_i&=\sum_{j\neq i}2\alpha_{ij}\alpha_{ji}(E_j-E_i).
\end{align}
Here, we only consider the leading term and therefore focus on the coefficient matrix for $\Delta t^3$. This result demonstrates that, no matter what non-degenerate Hamiltonian we choose, as long as $A$ and $B$ are not diagonal, we will have $\xi_1\equiv \xi_2$ to the leading order. This means that, by averaging these two product formulas, we can never reduce the leading-order ($\propto \Delta t^3$) coefficient in $\xi$. All the error reduction in this case will be from the reduction in the first term of \cref{eq:B1} by combining $[H,\eta_1]$ and $[H,\eta_2]$ for short-time simulations. 

Suppose that the optimal weight $p$ for the second-order simulation is $p_{\text{opt}}$ and that it results in a reduction of the term $[H,\eta]$.  
If we now consider fourth-order simulation, Eq.~\eqref{eq:STHighApprSimp} implies that the resulting error does not depend on the commuting part ($\xi$) of the second-order error, and is instead entirely determined by the non-commuting part ($[H,\eta]$) of the second-order error. Therefore, the optimal weight $p$ for the fourth-order simulation is the same as the optimal weight $p_{\text{opt}}$ for the second-order simulation, and the corresponding error is reduced by the same factor as the non-commuting part of the second-order error. Since the second-order error also contains the commuting part, which is not reduced, this concludes the proof of part (i) of Lemma~\ref{clm:ShortABErrRed}.

To prove part (ii) of Lemma~\ref{clm:ShortABErrRed}, we first derive expressions of the first term in the second-order error:
\begin{align}
[H,\eta_1]&\propto\sum_{i,j\neq i}[\alpha_{ij}(E_j-E_i)(2E_i-\alpha_i-2E_j+\alpha_j)+\sum_{l\neq i,j}\alpha_{il}\alpha_{jl}(2E_l-E_i-E_j)]\ket{i}\bra{j},\\
[H,\eta_2]&\propto\sum_{i,j\neq i}[\alpha_{ij}(E_j-E_i)(2E_j-\alpha_i-2E_i+\alpha_j)+\sum_{l\neq i,j}\alpha_{il}\alpha_{jl}(2E_l-E_i-E_j)]\ket{i}\bra{j}.
\end{align}
Again, we only consider the leading term and therefore focus on the coefficient matrix for $\Delta t^3$. We consider combining two product formulas of $(2k+4)$th order. By using \cref{eq:STHighApprSimp} $k$ times, we can approximate the simulation error for the two formulas as
\begin{align}
&[H,...,H,[H,\eta_1]]\propto\sum_{i,j\neq i}(E_j-E_i)^{2k}[\alpha_{ij}(E_j-E_i)(2E_i-\alpha_i-2E_j+\alpha_j)+\sum_{l\neq i,j}\alpha_{il}\alpha_{jl}(2E_l-E_i-E_j)]\ket{i}\bra{j},\\
&[H,...,H,[H,\eta_2]]\propto\sum_{i,j\neq i}(E_j-E_i)^{2k}[\alpha_{ij}(E_j-E_i)(2E_j-\alpha_i-2E_i+\alpha_j)+\sum_{l\neq i,j}\alpha_{il}\alpha_{jl}(2E_l-E_i-E_j)]\ket{i}\bra{j}.
\end{align}
As $k$ increases, only the terms with large $(E_j-E_i)$ will matter. At large $k$, the value of a few terms will determine the absolute value of the simulation error, which requires fewer weight parameters to optimize. This observation makes it possible to obtain significant error reduction via averaging over a few product formulas. In the limit $k \rightarrow \infty$, since the largest $(E_j-E_i)$ is unique as the Hamiltonian is non-degenerate, we can ignore all but one entry in $[H,...,H,[H,\eta_1]]$ and $[H,...,H,[H,\eta_2]]$. Therefore, the optimal weights for our approach would converge to a particular value $p^*$ that only depends on these two terms as the order $k$ increases. This concludes the proof of part (ii) of Lemma~\ref{clm:ShortABErrRed} in the main text. Furthermore, in the case when this entry has the opposite sign in the two errors, we can completely eliminate the leading-order error in this limit by adjusting $p$.

\section{Application to classical algorithms for studying  quantum models of infinite size}\label{section:iTEBD}
In this section, we consider applying our averaging approach to classical simulations of quantum models. 
The infinite time-evolving block decimation (iTEBD) algorithm \cite{Vidal2007Classical,Orus2008Infinite,Jordan2008Classical} is widely used to approximate ground-state energies and ground states of one-dimensional quantum models. It is a method based on matrix product states (MPS) \cite{Perez2006Matrix} and can approximate the ground state of an infinite-size one-dimensional chain. 

We consider a one-dimensional translationally-invariant infinite-chain model with sites labeled by $r\in\mathbb{Z}$ and each described by a complex $D_s$-dimensional vector space $V_r\cong\mathbb{C}^{D_s}$. 
Starting with a translationally invariant initial pure state $|\Psi\rangle$ not orthogonal to the desired ground state, we can find the ground state by evolving $|\Psi\rangle$ in imaginary time:
\begin{align}\label{eq:iTEBDG}
|G\rangle=\lim_{\tau\to\infty}e^{-\tau H}|\Psi\rangle,
\end{align}
which requires classically simulating $e^{-\tau H}$.

The iTEBD method represents $|\Psi\rangle$ using an MPS, and we briefly recap here the main idea of this representation. For any site $r$, we denote by $[\leftarrow r]$ and $[r+1\rightarrow]$ the semi-infinite sub-chain of sites $\{-\infty,...,r\}$ and $\{r+1,...,\infty\}$, respectively. We then write the Schmidt decomposition of $\vert\Psi\rangle$ across this partition as
\begin{align}\label{eq:SchDecomp}
|\Psi\rangle=\sum_{\alpha=1}^{D_p}G_{\alpha}^{[r]}|\Psi_\alpha^{[\leftarrow r]}\rangle\otimes|\Psi_\alpha^{[r+1\rightarrow]}\rangle,
\end{align}
where the Schmidt rank $D_p$ is assumed to be finite, $\{G_\alpha^{[r]}\}$ are the Schmidt coefficients, and $\{|\Psi_\alpha^{[\leftarrow r]}\},\{|\Psi_\alpha^{[r+1\rightarrow]}\}$ are basis states. We use a three-index tensor $T^{[r]}$ to relate the Schmidt basis across two neighboring partitions.
Therefore, the initial state $|\Psi\rangle$ can be expanded using the local basis $|i^{[r]}\rangle$ for site $r$ and using $G^{[r]}T^{[r+1]}G^{[r+1]}$:
\begin{align}\label{eq:MPSDecomp}
|\Psi\rangle=\sum_{\alpha,\beta=1}^{D_p}\sum_{i=1}^{D_s}G_\alpha^{[r]}T^{[r+1]}_{i\alpha\beta}G_\beta^{[r+1]}|\Psi_\alpha^{[\leftarrow r]}\rangle|i^{[r]}\rangle|\Psi_\alpha^{[r+2\rightarrow]}\rangle.
\end{align}
We can then decompose $|\Psi\rangle$ using local bases for sites $\{r+1,r+2\}$ in terms of $G^{[r]}T^{[r+1]}G^{[r+1]}T^{[r+2]}G^{[r+2]}$ and so on. By repeatedly exploiting this decomposition, we can derive an MPS representation for $|\Psi\rangle$.
\begin{figure}
    \centering
    \includegraphics[width=0.8\textwidth]{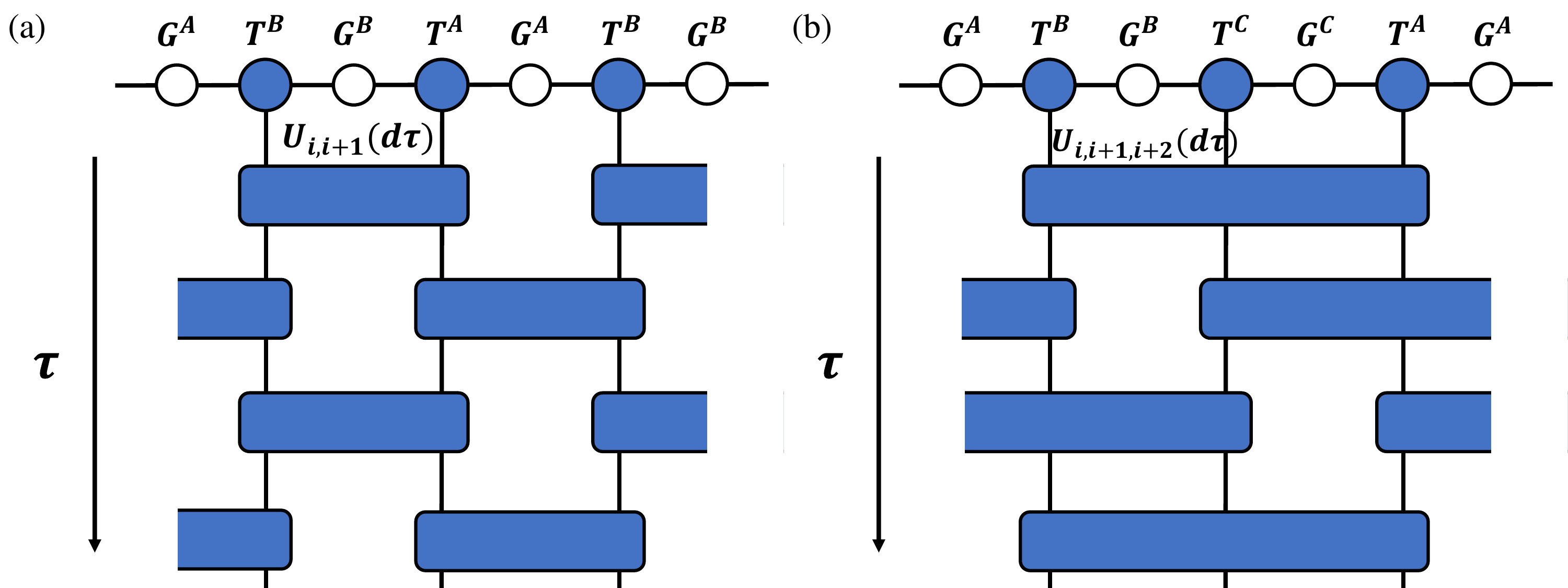}
    \caption{An illustration of (a) $2$-site and (b) $3$-site iTEBD algorithm using the first-order product formula.}
    \label{fig:iTEBDMPS}
\end{figure}

Notice that directly implementing the evolution using the entire Hamiltonian $H$ is impossible for MPS. We thus implement the imaginary-time evolution $e^{-\tau H}$ using digital simulation. A direct method is to split the Hamiltonian and apply the first-order product formula. As shown in Fig.~\ref{fig:iTEBDMPS}(a), in case of a two-body nearest-neighbor Hamiltonian $H$, we can split $H$ into a block of terms $H_{i,i+1}$ with even $i$ ($H_{\text{even}}$) and a block of terms $H_{i,i+1}$ with odd $i$ ($H_{\text{odd}}$). Therefore, for a fixed step size $d\tau$, the digital simulation of $e^{-\tau H}$ at $\tau= Nd\tau\to\infty$ takes the form
\begin{align}\label{eq:2-siteiTEBD}
\lim_{\tau\to\infty}e^{-\tau H}\approx\prod_{N\to\infty}\left(\prod_{i=\text{odd}}U_{i,i+1}(d\tau)\right)\left(\prod_{i=\text{even}}U_{i,i+1}(d\tau)\right)=\prod_{N\to\infty}\left(\prod_{i=\text{odd}}e^{-d\tau H_{i,i+1}}\right)\left(\prod_{i=\text{even}}e^{-d\tau H_{i,i+1}}\right),
\end{align}
where $U_{i,i+1}(d\tau)=e^{-d\tau H_{i,i+1}}$ and $N$ is the number of steps for the digital simulation. This method is known as a $2$-site 
iTEBD method. In each iteration corresponding to time step $d\tau$, we contract  $G^{[r]}T^{[r+1]}G^{[r+1]}$ and apply a block $U_{i,i+1}(d\tau)$. We then use singular-value decomposition to decompose the tensor into the new $G^{[r]}T^{[r+1]}G^{[r+1]}$.

Similarly, a Hamiltonian composed of terms $H_{i,i+1,i+2}$ acting on triples of adjacent sites can be decomposed into three parts according to the value of $i$ modulo $3$, as shown in Fig.~\ref{fig:iTEBDMPS}(b). This method is known as the $3$-site iTEBD method.
Compared to the $2$-site algorithm, two singular-value decompositions are required to decompose the tensor into a product of five smaller tensors. 
While we presented the example of a first-order product formula, imaginary time evolution can also be realized using higher-order product formulas. 

Now we consider constructing NUSCs by averaging different Suzuki-Trotter product formulas. When calculating the ground-state energy via $\Tr(H|G\rangle\langle G|)/\Tr(|G\rangle\langle G|)$ with $|G\rangle=\lim_{\tau\to\infty}e^{-\tau H}|\Psi\rangle$, we can represent $|G\rangle\langle G|$ as a weighted average of different simulations $U_1,...,U_M$ as
\begin{align}\label{eq:comiTEBD}
|G\rangle\langle G|=\sum_{m=1}^M p_m U_m|\Psi\rangle\langle\Psi|U_m^\dagger,
\end{align}
where $\{p_1,...,p_M\}$ are the weights. Since the iTEBD algorithm is a classical computational method, we can directly take the average of the unitaries $U_m$, which can be regarded as a non-unitary operation on the initial state $|\Psi\rangle$:
\begin{align}\label{eq:comiTEBDNU}
|G\rangle=\sum_{m=1}^Mp_mU_m|\Psi\rangle.
\end{align}

To study the performance of the averaging technique, we perform a numerical experiment on both two-body Hamiltonians and three-body Hamiltonians. For the $2$-site iTEBD method, we consider a Heisenberg model with nearest-neighbor $2$-body terms:
\begin{align}\label{eq:HeisenbergLoc}
H=\sum_{i}\left(X_iX_{i+1}+Y_{i}Y_{i+1}+Z_iZ_{i+1}\right),
\end{align}
where $X_i,Y_i, Z_i$ are Pauli operators on spin $i$. We exploit an algorithm similar to that in Ref.~\cite{Vidal2007Classical} and decrease time-step size $d\tau\in\{0.1,0.01,0.001\}$. 
As we expect the ground state as the output of the algorithm, the converged state should remain unchanged after each iteration $d\tau$. We measure the distance between the states before and after each iteration and regard it as the error for each iteration. Once this distance is below a certain threshold, we say the algorithm has converged at the current time-step size and we decrease the time-step size. Otherwise, we enter the next iteration. We count the total number of iterations required for the algorithm to converge for all $d\tau\in\{0.1,0.01,0.001\}$. We set the threshold value as $10^{-10}$ in the following numerical experiments. 
As shown in Fig.~\ref{fig:iTEBD}(a), we compare the number of iterations that the algorithm requires to convergence for first, second, and fourth-order standard Suzuki-Trotter formulas and for the unweighted average ($p = 0.5$)---using \cref{eq:comiTEBD}---of two first- and second-order Suzuki-Trotter formulas corresponding to the two orderings of $H_{\text{odd}}$ and $H_{\text{even}}$. In Fig.~\ref{fig:iTEBD}(b), we perform a similar numerical experiment using a linear combination of unitaries [\cref{eq:comiTEBDNU}] instead of taking a mixture of states [\cref{eq:comiTEBD}]. We observe that, by averaging different simulations, one can reduce the required iteration number by a factor of three.
\begin{figure}
    \centering
    \includegraphics[width=0.99\textwidth]{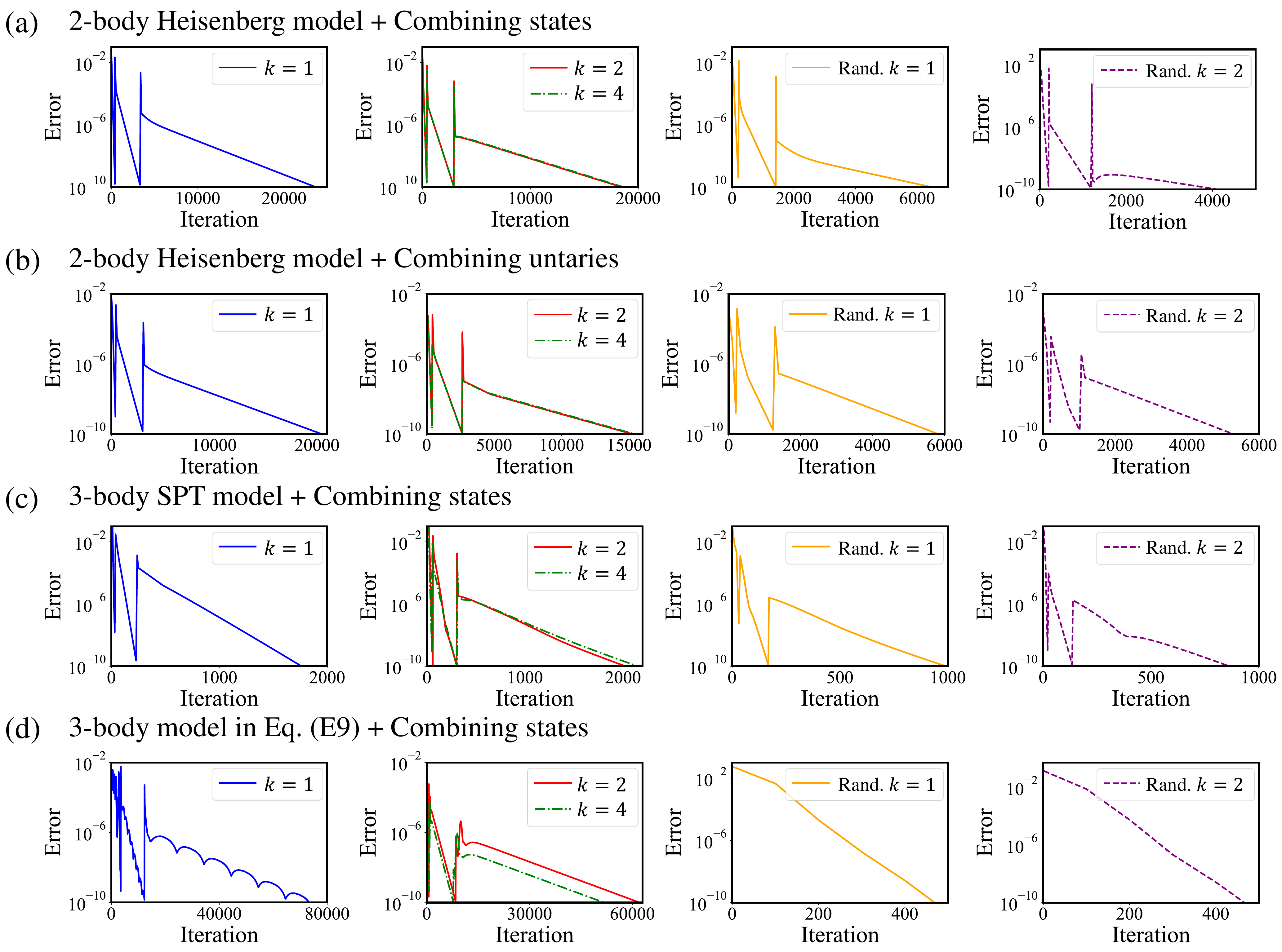}
    \caption{Application of averaging simulations in iTEBD methods with equal weight $p_1=p_2=0.5$ for two first- and second-order Suzuki-Trotter formulas with permuted sequences.} (a) Using $2$-site iTEBD methods to calculate the ground state of the Hamiltonian in Eq.~\eqref{eq:HeisenbergLoc}. We compute the error as a function of iteration number for first- ($k=1$), second- ($k=2$), and fourth-order ($k=4$) Suzuki-Trotter formulas and for combining two first-order (labeled ``Rand.~$k=1$'') and two second-order (labeled ``Rand.~$k=2$'') simulations via Eq.~\eqref{eq:comiTEBD}. (b) Same as (a), but using a linear combination of unitaries [\cref{eq:comiTEBDNU}].  
    (c) Same as (a), but using the $3$-site iTEBD method to calculate the ground state of the Hamiltonian in Eq.~(\ref{eq:SPT}) with $J_i=h_i=-V_i=1$. 
    (d) Same as (c), but for the Hamiltonian in \cref{eq:1Dspin}.
    \label{fig:iTEBD}
\end{figure}

For the $3$-site iTEBD method, we observe a similar speedup in the convergence of the iTEBD algorithm. Specifically, we consider the family of Hamiltonians
\begin{align}\label{eq:SPT}
H=-\sum_i(J_iZ_{i-1}X_iZ_{i+1}+V_iX_iX_{i+1}+h_iX_i),
\end{align}
specified by parameters $J_i$, $V_i$, and $h_i$. This family of Hamiltonians is used in the study of topological edge states defined at arbitrarily high energies \cite{Bahri2015Localization} and in the study of Floquet symmetry-protected topological (SPT) phases  \cite{Deng2021Observation,Jeyaretnam2021Quantum}. The Hamiltonian in \cref{eq:SPT} has a $\mathbb{Z}_2\times\mathbb{Z}_2$ symmetry. In the extreme case when $V_i=h_i=0$, the model is exactly solvable by mapping to free fermions, and the eigenstates of \cref{eq:SPT} are the mutual eigenstates of the stabilizers $Z_{i-1}X_iZ_{i+1}$. These eigenstates are SPT states and the ground states are called cluster states. When $V_i,h_i\neq 0$, the single- and two-body terms make the ground states deviate from cluster states, and the model can not be solved exactly. In Fig.~\ref{fig:iTEBD}(c), we fix the parameters to be $J_i=h_i=-V_i=1$ and apply the $3$-site iTEBD method to calculate the ground-state energy and the ground-state wavefunction. In particular, we plot the error for the ground state (the trace distance between the output state and the actual ground state) as a function of the number the iterations. In order to perform averaging, we combine with equal weight ($p=0.5$) via Eq.~(\ref{eq:comiTEBD}) two product formulas, such that the first product formula has the term ordering as in \cref{eq:SPT} and the second product formula has the order of the terms reversed. We see that, using the averaging technique, we achieve a $50\%$ reduction in the iteration number required to converge to the ground state. 

It is worth mentioning that, in some cases, we can get more dramatic reductions in the number of iterations required to converge.  As an example, consider the following one-dimensional spin model with $3$-body terms:
\begin{align}\label{eq:1Dspin}
H=-\sum_i(Z_{i-1}X_{i}Z_{i+1}+Y_i).
\end{align}
In Fig.~\ref{fig:iTEBD}(d), we use the $3$-site iTEBD method to approximate the ground state and the ground-state energy for this model. For the averaging technique, we again combine with equal weight ($p=0.5$) via Eq.~(\ref{eq:comiTEBD}) two product formulas at both order $k=1$ and order $k=2$ with the second one reversing the order of the three Hamiltonian terms. While the algorithms that exploit standard (not averaged) Suzuki-Trotter formulas require more than $5\times 10^4$ steps even for fourth-order formulas, a simple equally-weighted combination of two first-order product formulas requires fewer than $500$ steps to converge to the final energy, which is a two-orders-of-magnitude improvement compared to the unaveraged algorithm. 

\section{Derivation of step and sample complexity for long-time simulations}\label{section:ConvRandSam}
In \Cref{sec:TermOrdering}, we have shown that, by applying permutations of Hamiltonian terms and using weighted averaging, we can reduce the quantum simulation error. 
In particular, we demonstrated how this approach works in the case where we choose only a few permutations. By choosing a larger number of contributing permutations, we expect an increasing ability to reduce the simulation error.  
In particular,  in \Cref{sec:TermOrdering}, we show that averaging over the complete permutation group with equal weights for a Hamiltonian composed of $\Gamma$ terms provides a reduction in the asymptotic leading error term by a factor of $\Gamma$. This approach is immediately applicable to Hamiltonians with a small number of terms. However, the size of the permutation group grows exponentially with the number of Hamiltonian terms, and averaging a large number of permuted product formulas requires more resources. To address this issue, we mention in \Cref{sec:TermOrdering} an alternative approach: to randomly sample all the product formulas in the permutation group with equal weights. In \Cref{sec:TermOrdering}, we also give an upper bound on the statistical error from sampling a limited number of permutations. In this appendix, we prove this bound. The resulting bound also holds for \Cref{sec:Sym}, where we are averaging over a large number of symmetries in the symmetry group with equal weights.  For technical simplicity, we consider the spectral norm in this section, which is different from the Frobenius norm considered in the main text. 

We assume that the simulation $U_m(N,\Delta t)$ is an $N$-step simulation with the simulation block $U_m(\Delta t)$ in each step being a product formula with leading-order error $O(t^q)$. Consider running our algorithm $T$ times (corresponding to $T$ different choices of $m$ in $U_m(N,\Delta t)$) and averaging (with equal weights) the results of observable measurements afterward. The step size should be small, $\Delta t=t/N<<1$, while $t$ could have an arbitrary value. The simulation error for a particular input state $\rho$ can be written as $ (1/T) \sum_j U_{(j)} \rho U^\dagger_{(j)}-V \rho V^\dagger$, where $V$ is the ideal evolution. Defining $\mathcal{U}(\rho)=(1/T) \sum_j U_{(j)} \rho U^\dagger_{(j)}$ and $\mathcal{V}(\rho)=V\rho V^\dagger$, we can write the spectral norm of the simulation error as
\begin{align}
\norm{\mathcal{U}(\rho)-\mathcal{V}(\rho)}&\leq\norm{\mathcal{U}(\rho)-\mathcal{V}(\rho)}_1\nonumber\\
&\leq\max_{\rho}\norm{\mathcal{U}(\rho)-\mathcal{V}(\rho)}_1\nonumber\nonumber\\
&\leq\max_{\varphi\in\mathcal{C}^{4^n\times 4^n}}\norm{(\mathbbm{I}\otimes\mathcal{U})(\varphi)-(\mathbbm{I}\otimes\mathcal{V})(\varphi)}_1\nonumber\\
&=\norm{\mathcal{U}-\mathcal{V}}_{\diamond}\nonumber\\
&\leq 2\norm{\frac{1}{T}\sum_j U_{(j)}-V}
\end{align}
where $\norm{\cdot}_1$ is the trace norm, $\norm{\cdot}_{\diamond}$ is the diamond norm between quantum channels, and $\mathbbm{I}\otimes\mathcal{U}(\cdot)$ is the channel that performs the identity channel on the first $n$ qubits and $\mathcal{U}(\cdot)$ on the remaining $n$ qubits. Here, the first line follows from the definition of the norms, the third line uses the fact that $\norm{(\mathbbm{I}\otimes\mathcal{U})\left(\frac{I}{2^n}\otimes\rho\right)-(\mathbbm{I}\otimes\mathcal{V})\left(\frac{I}{2^n}\otimes\rho\right)}_1=\norm{\mathcal{U}(\rho)-\mathcal{V}(\rho)}_1$ for any $\rho$, and the last line follows from Lemma~3.4 of Ref.~\cite{Chen2021ConcentrationPRX}. We decompose $(1/T) \sum_j U_{(j)}-V$ into the expectation bias $\mathbb{E}[U_{(j)}]-V$ and the fluctuation error represented as
\begin{equation}\label{eq:SampNErr}
E_{\text{Fluc}}=\frac 1T\sum_{j=1}^TU_{(j)}(N,\Delta t)-\mathbb{E}\left[U_m(N,\Delta t)\right],
\end{equation}
where $U_{(j)}$ is the $j$-th random sample of our algorithm (i.e.\ one of the $U_m(\Delta t)$). The following theorem asserts the sample ($T$) and step ($N$) complexity required to bound $\norm{E_{\text{Fluc}}}_2$. We remark that this sample complexity does not consider the number of samples that are required to converge to expectation values of observables. 

\begin{theorem}\label{thm:H}
Consider simulating the dynamics $V=e^{-iHt}$ of an $n$-qubit Hamiltonian $H=\sum_{j=1}^\Gamma H_j$. 
Pick $\epsilon, \delta > 0$ and consider drawing $T$ random simulations $U_{(1)},\ldots,U_{(j)},\ldots,U_{(T)}$ each of $N$ steps and $\Delta t=t/N\ll 1$ such that
\begin{equation}\label{eq:SamComp}
T = \Omega \left(\frac{\gamma t^q}{N^{(2q-1)/2}\epsilon}\left(n+\log\left(\frac1\delta\right)\right)^{1/2}\right),
\end{equation}
where $\gamma = \norm{E_{m}^{(q)}}_2+\norm{\xi_{m}^{(q)}}_2$, $E_{m}^{(q)}$ is the error matrix for each time step, and $\xi_{m}^{(q)}$ is the part of $E_{m}^{(q)}$ that commutes with $H$, following notation from the main text. Then with probability at least $1-\delta$, the fluctuation error $||E_{\text{Fluc}}||_2$ is bounded above by $\epsilon$. Equivalently, if we fix $T$ and set the step number to be
\begin{equation}\label{eq:StepComp}
N= \Omega  \left(\frac{\gamma^{2/(2q-1)}t^{2q/(2q-1)}}{(T\epsilon)^{2/(2q-1)}}\left(n+\log\left(\frac1\delta\right)\right)^{1/(2q-1)}\right),
\end{equation}
then the fluctuation error $\norm{E_{\text{Fluc}}}_2$ is bounded above by $\epsilon$ with probability at least $1-\delta$.
\end{theorem}

\begin{proof}
The proof of \Cref{thm:H} follows from the matrix Bernstein inequality  \cite{Tropp2015Introduction,Tropp2011Freedman,Gross2011Recovering}. We introduce several mathematical results regarding matrix concentration inequalities. Suppose $X_1,...,X_N$ are independent and identically distributed (i.i.d.) random variables. According to the strong law of large numbers, the sample mean $1/N\sum_{i=1}^NX_i$ converges to the expectation $\mathbb{E}[X_i]$ for large $N$. For random matrices $X_1,...,X_N$, the following matrix Bernstein inequality holds \cite{Tropp2015Introduction}:

\begin{lemma}[matrix Bernstein inequality~\cite{Tropp2015Introduction,Gross2011Recovering,Oliveira2009Concentration,Pinelis1999Optimum,Tropp2011Freedman}] Consider a set of $d\times d$ random matrices $\{X_1,...,X_T\}$ with $\mathbb{E}[X_i]=0$ and $||X_i||_2\leq R$ for all $i$. Then for any $\tau>0$, we have
\begin{equation}\label{eq:bernstein}
\textbf{Pr}\left[\norm{\sum_{i=1}^T X_i}_2\geq \tau\right]\leq 2d\exp\left(\frac{-\tau^2/2}{TR^2+R\tau/3}\right). 
\end{equation}
\end{lemma}

Now we continue to prove \Cref{thm:H}. We first prove the sample complexity result given in \cref{eq:SamComp}.  We denote $X_j=\frac 1T U_{(j)}-\frac 1T(\mathbb{E}[U_m])$ for $j=1,\ldots,T$, which satisfies $\mathbb{E}[X_j]=0$. Moreover, the spectral norm of each random variable $X_j$ is bounded by
\begin{align}\label{eq:boundX}
\norm{X_j}_2&=\frac{1}{T}\norm{U_{(j)}-\left(\mathbb{E}[U_m]\right)}_2\\
&\leq\frac{1}{T}\norm{\mathbb{E}[U_m]-V}_2+\norm{V-U_{(j)}}_2\\
&\leq \frac{N}{T}\max_{m}\left[\norm{E_{m}^{(q)}}_2+\norm{\xi_{m}^{(q)}}_2\right]t^q/N^q.
\end{align}
In the second line, we used the triangle inequality. The first term in the second line is bounded by $N\norm{\xi_m^{(q)}}_2(t/N)^q$ according to Eq.~\eqref{eq:LongUmExpansion} in the main text if we only consider the leading-order error. The second term in the second line is bounded by $N\norm{E_m^{(q)}}_2$ using the assumption that $U_m$ is a product formula with leading order $O(t^q)$. Thus we obtain the bound in the third line. We apply the matrix Bernstein inequality to the sum $\frac 1T\sum_{j=1}^TU_{(j)}-(\mathbb{E}[U_m])=\sum_{j=1}^T X_j$ with $R=\gamma t^q/TN^{q-1}$. Therefore, given $\epsilon\leq NR$, the probability that $\norm{\frac 1T\sum_{j=1}^TU_{(j)}-(\mathbb{E}[U_m])}_2\geq\epsilon$ can be bounded by
\begin{align}\label{eq:boundSamp}
\textbf{Pr}\left[\norm{\sum_{i=1}^T X_i}_2\geq \epsilon\right]&=\textbf{Pr}\left[\norm{\frac 1T\sum_{j=1}^TU^{N,j}...U^{1,j}-(\mathbb{E}[U_m])^N}_2\geq\epsilon\right]\\
&\leq 2d\exp\left(\frac{-\epsilon^2/2}{NR^2+R\epsilon/3}\right)\\
&\leq 2d\exp\left(\frac{-3\epsilon^2}{8NR^2}\right)\\\label{eq:boundSampN}
&=2d\exp\left(\frac{-3T^2N^{2q-1}\epsilon^2}{8\gamma^2t^{2q}}\right).
\end{align}
Therefore, if we choose $T\geq\Theta \left(t^q\gamma\epsilon^{-1}N^{-(2q-1)/2}\left(n+\log( 1/\delta)^{1/2}\right)\right)$, then we have $\norm{\frac 1T\sum_{j=1}^TU_{(j)}-(\mathbb{E}[U_m])}_2\leq\epsilon$ with probability at least $1-\delta$. 

For the step complexity when we fix $T$, it follows from Eq.~\eqref{eq:boundSampN} that, if $N\geq\Theta\left(\gamma^{2/(2q-1)}t^{2q/(2q-1)}\left(n+\log(1/\delta)\right)^{1/(2q-1)}/(T\epsilon)^{2/(2q-1)}\right)$, then $\norm{\frac 1T\sum_{j=1}^TU_{(j)}-(\mathbb{E}[U_m])}_2\leq\epsilon$ with probability at least $1-\delta$.
\end{proof} 

The sample complexity given by  Eq.~\eqref{eq:SamComp} has polynomial dependence on $t$ and $n$. Moreover, when the step size $\Delta t=t/N$ is fixed, the sample complexity is $\Theta\left(N^{1/2}\gamma\Delta t^{q}\left(n+\log(1/\delta)\right)^{1/2}/\epsilon\right)$. This means that, when step size $\Delta t$ is fixed, the sample complexity scales as $T\propto t^{1/2}$. On the other hand, if we replace $\Delta t$ with $t/N$, we find that sample complexity scales as $N^{-(2q-1)/2}$. This means that, for fixed $t$, we can reduce sample complexity by increasing the step number $N$ (i.e.\ by decreasing the step size $\Delta t$). Alternatively, we can also bound the fluctuation error with high probability even with $T=1$ random sequence (i.e.\ no averaging), provided that the step complexity $N$ is large enough to satisfy Eq.~\eqref{eq:StepComp}. In this extreme case, the contributing simulation error itself is very small, which makes the fluctuation error bounded by $\epsilon$. Compared with the step complexity of qDRIFT, which is $\Theta\left(t^2\gamma^2\left(n+\log(1/\delta)\right)/\epsilon^2\right)$ \cite{Chen2021ConcentrationPRX}, our step complexity has more favorable dependence on $t,\epsilon,\gamma,\left(n+\log(1/\delta)\right)$ as we employ higher-order ($q > 1$) digital quantum simulations in each step. In particular, when the simulations $U_{(1)},...,U_{(T)}$ are chosen as product formulas with large $q$, our step complexity approaches $\Theta\left(\gamma^{1/q}t\left(n+\log(1/\delta)\right)^{1/(2q)}/\epsilon^{1/q}\right)$, which is linear in $t$. It is worth mentioning that this algorithm is scalable to larger systems at the price of an $\Theta\left(n^{1/2}\right)$ increase in sample complexity or an $\Theta\left(n^{1/(2q)}\right)$ increase in step complexity. Notice that, since step complexity $N$ scales as $T^{-2/(2q-1)}$, increasing sample number $T$ is not an effective way to reduce step complexity at large $q$. Finally, we remark again that the sample complexity considered in this paper doesn't take into account the experimental fact that we have to repeat the experiment multiple times to get expectation values of observables.

\section{Symmetry-based error reduction for the long-range Heisenberg chain}\label{section:Equiv}
\begin{figure}
    \centering
    \includegraphics[width=0.99\textwidth]{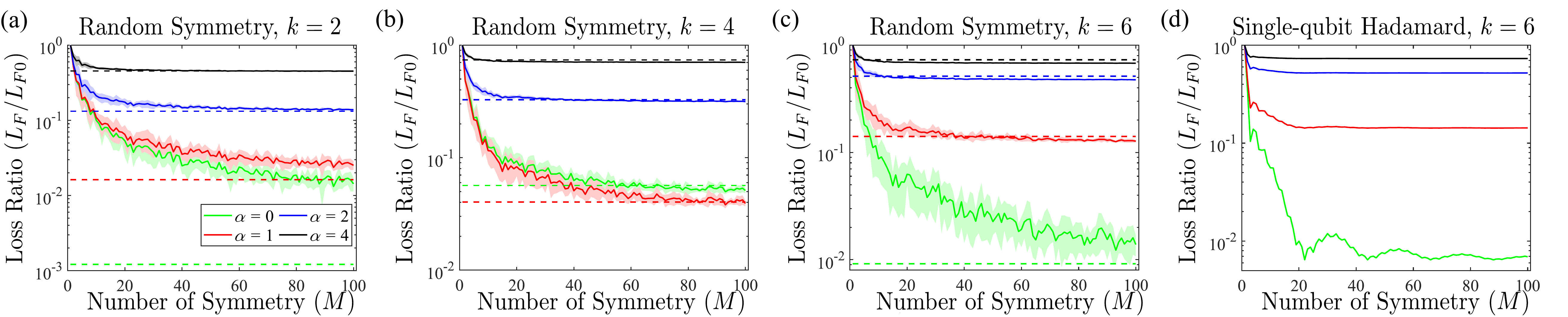}
    \caption{
    (a) We employ the second-order ($k=2$) product formula as $U_1$ and plot the error reduction as a function of the number of Haar-randomly chosen symmetry transformations at $t=100$ and different power-law exponents $\alpha=0,1,2,4$. The dashed horizontal line represents the error reduction achieved by combining $3!$ second-order product formulas with all possible orderings of Hamiltonian terms  $H_X$, $H_Y$, and $H_Z$. (b) Same as (a), but for the $4$th-order ($k=4$) product formula. (c) Same as (a), but for the $6$th-order ($k=6$) product formula. (d) We use the symmetry operator $\sum_{i=1}^n H_i$ and construct $C_m=\text{exp}(i(m) \bigoplus_{i=1}^n H_i\Delta)$ for $m=0,...,M-1$ and $\Delta=0.01$, where $H_i$ is the Hermitian Hadamard matrix acting on qubit $i$. The initial simulation is the $6$th-order product formula. We plot the error reduction as a function of the number of symmetry transformations at different power-law exponents $\alpha$.}
    \label{fig:symnum}
\end{figure}
In this section, we first provide additional numerical results for the symmetry protection approach in \Cref{sec:Sym}. We then prove Eq.\ (\ref{eq:SymReduce}) in the main text and provide numerics illustrating this equation.

In Fig.~\ref{fig:symnum}(a,b,c), we consider the Heisenberg spin chain with power-law interactions [see Eq.\ (\ref{eq:PowerHeisenberg}) in the main text]. We choose symmetry transformations Haar-randomly and observe that, as the number of symmetry transformations $M$ increases, the final error reduction (solid lines) approaches an asymptotic value (dashed lines) given by the uniform mixture of $3!$ $k$th-order Suzuki-Trotter formulas with all possible term orderings of $H_X$, $H_Y$, and $H_Z$. 
For the case of $k=1$, the leading-order (second-order) error only contains commutators $[H_X,H_Y]$, $[H_X,H_Z]$, and $[H_Y,H_Z]$. By averaging over all $6$ sequences, we can eliminate the second-order error term. Averaging over Haar-randomly chosen symmetry transformations also reduces the leading-order error to zero as the commutator of the form $[H_X,H_Y]$ contains no quadratic terms such as $X^2$ where $X=\sum_{i=1}^n X_i$.
Yet, we were not able to rigorously prove the equivalence between the two approaches in the general case for $k\geq 2$ due to the non-isotropic structure of the high-order error terms.

We now prove Eq.\ (\ref{eq:SymReduce}) in the main text and provide numerics illustrating this equation. We consider a specific choice of the Hamiltonian-like operator $O=\sum_{i=1}^n H_i$, where $H_i$ is the Hadamard matrix acting on qubit $i$, that generates global rotations around the axis defined by $\mathbf{n}=\frac{1}{\sqrt{2}}(1,0,1)$.  Constructing the group elements as $C_m=\text{exp}(i m O\Delta)$ for $m=0,...,M-1$ such that $\norm{O}_F\Delta\ll 1$ and at least on eigenvalue of $O \Delta$ is an irrational multiple of $\pi$, we see that the resulting group has infinite size since the rotation angles $\phi_m$ are irrational multiples of $\pi$. In a generic case, %when the symmetry transformation $e^{i O}$ does not correspond to a permutation of Hamiltonian terms ($X\leftrightarrow Y\leftrightarrow Z$), and 
when the symmetry group is infinite (the phases are irrational multiples of $\pi$), it is beneficial to have large $M$ because this reduces the error. For an infinite symmetry group $\{C_m=\text{exp}(i m O\Delta)\}_{m=0}^{\infty}$ %such that $e^{i O}$ does not correspond to a permutation of Hamiltonian terms 
and for a given simulation error $E$, we decompose the error as $E=[O,\eta_C]+\xi_C$ similar to \cref{eq:ErrDecomp} in the main text, where $[O,\eta_C]$ ($\xi_C$) is the part of the error that does not commute (commutes) with $O$. We have
\begin{align}
E_q^{\text{sym}}
&=\frac{1}{M}\sum_{i=0}^{M-1} e^{-imO\Delta}Ee^{i(m-1)O\Delta}
\approx\frac{1}{M}\int_{0}^{(M-1)\Delta}e^{-iOt'}\frac{E}{\Delta}e^{iOt'}dt'\\
&=\xi_C+\frac{1}{M\Delta}\left[e^{-i(M-1)O\Delta}\eta_Ce^{i(M-1)O\Delta}-\eta_C\right],
\end{align}
where $E_q^{\text{sym}}$ is defined in \cref{eq:SymProbEq} in the main text. This provides the proof for \cref{eq:SymReduce} in the main text, which indicates the $O(1/M)$ reduction of $\eta_C$, which is the part of the error that does not commute of $O$. This allows us to consider the unitary error term $E_q^{\text{sym}}$ as a function of the number of symmetry transformations $M$.
Fig.~\ref{fig:symnum}(d) shows that the NUSC error is suppressed as $M$ grows, similar to the Haar-random case in Fig.~\ref{fig:symnum}(a-c) and approaches an asymptotic value. 
Fig.~\ref{fig:symnum}(d) demonstrates that, if the symmetry group does not coincide with the permutation group of Hamiltonian terms, it is beneficial to average the NUSC over a large number $M$ of symmetry transformations.

\end{document}